\documentclass[pdflatex,sn-mathphys-num]{sn-jnl}

\usepackage{graphicx}%
\usepackage{multirow}%
\usepackage{amsmath,amssymb,amsfonts}%
\usepackage{amsthm}%
\usepackage{mathrsfs}%
\usepackage[title]{appendix}%
\usepackage{xcolor}%
\usepackage{textcomp}%
\usepackage{manyfoot}%
\usepackage{booktabs}%
\usepackage{algorithm}%
\usepackage{algorithmicx}%
\usepackage{algpseudocode}%
\usepackage{listings}%
\usepackage{placeins}%

\theoremstyle{thmstyleone}%
\newtheorem{theorem}{Theorem}%
\newtheorem{corollary}[theorem]{Corollary}%

\theoremstyle{thmstyletwo}%

\theoremstyle{thmstylethree}%

\begin{document}

\title[Modulation theory of KP solitons]{Modulation theory of Y- and X-wave Kadomtsev--Petviashvili solitons}

\author*[1]{\fnm{Jerry} \sur{Wang}}\email{jerry.wang@colorado.edu}
\author*[1]{\fnm{Mark A.} \sur{Hoefer}}\email{hoefer@colorado.edu}

\affil*[1]{\orgdiv{Department of Applied Mathematics}, \orgname{University of Colorado}, \orgaddress{\city{Boulder}, \state{Colorado}, \country{USA}}}

\abstract{
The evolution of long wavelength perturbations of traveling wave 2-soliton
  solutions of the Kadomtsev-Petviashvili II (KP) equation are
  investigated using modulation theory.  An initial-boundary value problem for a system of hyperbolic conservation laws is formulated for amplitude and slope perturbations to soliton legs of the resonant (Y-shaped) and X-shaped KP traveling wave solutions. The vertex where individuated soliton legs coalesce is modeled as a moving free boundary where information between the soliton legs can pass. The linearized hyperbolic system is proven to be well-posed under suitable initial conditions and is solved exactly. In the case of the resonant Y-shaped soliton, all admissible simple (rarefaction) waves of the nonlinear system are characterized. Specific perturbations of Y-waves can lead to the dimensionality of the modulated solution changing, corresponding to a physical geometric change to the multi-soliton pattern. Spontaneous leg deletion and leg creation can occur. The conditions under which the Y-wave modulation system can experience gradient catastrophe are given.  
  Gradient catastrophe is regularized by the formation of a new leg,
  with an augmented modulation system proposed to describe the growth of the new leg. On the other hand, the X-wave system can be reduced to two independent line soliton modulation systems. These findings reinforce the prevalence and generality of resonant Y-shaped vertices in the KP II system and applications.The theoretical modulation predictions are supported by direct numerical simulations of KP II, which show good agreement with the analysis.
}

\keywords{Kadomtsev--Petviashvili equation, KP II, soliton modulation, Y-solitons, X-solitons, rarefaction waves, shock formation}

\maketitle
\section{Introduction}
A fundamental problem of multidimensional nonlinear wave theory is the oblique interaction of solitary waves.  This problem has been studied using the multi-soliton solutions of the Kadomtsev–Petviashvili (KP) equation \cite{miles_obliquely_1977,biondini_line_2007,chakravarty_soliton_2009,kodama_kp_2010}, a universal model of weakly nonlinear and dispersive waves with weak transverse variation \cite{kadomtsev_stability_1970}.   More recent work has introduced the alternative approach of soliton modulation in which the interactions are described locally by line soliton (one-soliton) solutions subject to transverse modulation \cite{ryskamp_evolution_2021,ryskamp_modulation_2022}. The modulation equations are a system of hyperbolic conservation laws that asymptotically describe gradual changes in the line soliton's parameters, e.g., its amplitude and slope.  Expansive modulation solutions correspond to curved soliton fronts whereas compressive modulations result in modulation shock solutions that describe two-soliton traveling wave solutions such as the resonant Y-soliton and non-resonant X-soliton.  Modulation theory allows for the solution of initial value problems consisting of partial or truncated solitonic data whereas exact solutions do not.  Follow-on work includes the solution of the Riemann problem for the KP soliton modulation equations in which two partial solitons with different parameters meet at the origin \cite{han_regularizations_2026}.  The modulation approach does not require the existence of multi-soliton solutions and therefore applies to non-integrable equations such as the Benney-Luke \cite{hu_obliquely_2025} and two-dimensional Benjamin-Ono \cite{hu_diffraction_2024} equations.

The aim of this paper is to consider modulation configurations consisting of multiple interaction types, both compressive and expansive waves, a class of \emph{multi-modulation solutions}.  By doing so, we are able to describe modulations of multi-soliton solutions of the KP equation.  The main idea is to treat modulation shock solutions---the Y- and X-solitons---as conjoined, semi-infinite modulated solitons with a boundary that represents the vertex, allowing for the transfer of information between the soliton legs.  Our analysis demonstrates that the initial-boundary value problem is well-posed and predicts the spontaneous removal or creation of one leg of the resonant Y-soliton.

The paper is organized as follows.  In Sec.~\ref{sec:background}, we review relevant exact solutions (line solitons, Y-solitons, and X-solitons) to the KP II equation along with modulation theory for line soliton solutions. In particular, special modulation rarefaction wave solutions are introduced.  In Sec.~\ref{sec:modul-y-solit},  we formulate modulations of the Y-soliton in terms of an initial-boundary value problem on each of the three Y-soliton legs with the vertex serving as the free boundary. Following that, we prove the well-posedness of the linearized initial-boundary value problem for Y-soliton modulations. In Sec.~\ref{sec:ySoliModulations}, we classify the dynamics for a slow or fast rarefaction wave originating on each of the Y-soliton legs.  In Sec.~\ref{sec:xSoliModulations}, we formulate and classify modulations for X-solitons, which we treat as essentially two independent line soliton modulation systems, i.e., the X-soliton vertex is trivial.  Lastly, numerical results and comparisons are discussed in Sec.~\ref{sec:direct-numer-simul} demonstrating the effectiveness of modulation theory followed by some concluding remarks in Sec.~\ref{sec:conclusions}. 

\section{Background}
\label{sec:background}

The KP equation is a two-dimensional generalization of the well-known one-dimensional Korteweg–de Vries (KdV) equation. Both KP and KdV can be derived as asymptotic limits from more fundamental fluid models such as the Euler equations. The KP equation has two dispersion curvature signs, either positive (KP I) or negative (KP II).  We will exclusively consider 
the KP II equation in the form
\begin{equation}
    (u_t + uu_x + u_{xxx})_x + u_{yy} = 0,
    \label{eq:KPII}
\end{equation}
with stable line soliton solutions \cite{kadomtsev_stability_1970,alexander_transverse_1997,mizumachi_stability_2012} in the form
\begin{equation}
    u(x, y, t) =a\operatorname{sech}^2\left(\sqrt{\frac{a}{12}}(x+qy-c t - x_0)\right), \quad c=\frac{a}{3}+q^2,
    \label{eq:line_soli}
\end{equation}
where $a>0$ is the amplitude parameter, $q = \tan \phi \in \mathbb{R}$ is the slope as measured by the counterclockwise angle $\phi$ from the positive $y$ axis, and $x_0$ is the $x$-intercept of the overall soliton phase at $t=0$.  See Fig.~\ref{fig:kp_line_soliton}.

\begin{figure}[h!]
    \centering
    \includegraphics[clip=true,trim=0 25 0 35]{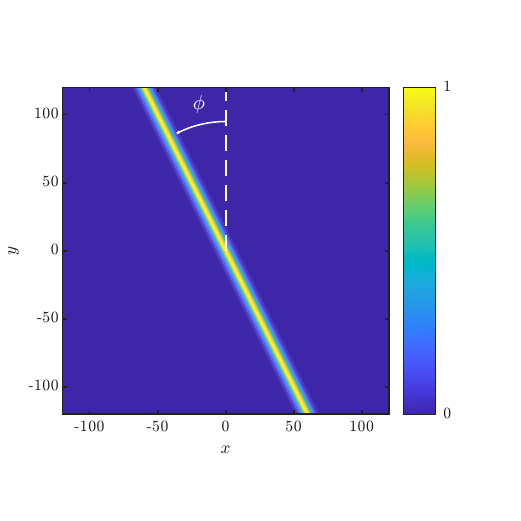}
    \caption{Example KP line soliton solution with $q = \tfrac12$ and $a = 1$.}
    \label{fig:kp_line_soliton}
\end{figure}

In addition to the line soliton solution, the KP II equation admits a family of multi-soliton solutions in the form \cite{biondini_soliton_2006,kodama_kp_2010}
\begin{equation}
    \label{eq:NSolitonSolution}
    u(x, y, t) = 12 \frac{\partial^2}{\partial x^2} \ln \tau(x, y, t),
\end{equation}
where $\tau(x, y,t)$ is defined as a Wronskian of exponentials,
\begin{gather*}
    \tau(x, y, t) = Wr(f_1, ..., f_N), \quad f_n(x, y, t) = \sum_{m=1}^M A_{n,m} e^{\theta_m (x, y, z)},\\
    \theta_m(x, y, t) = \frac{1}{\sqrt{12}} \left( k_m x + \frac{k_m^2}{2} y - \frac{k_m^3}{3} t\right),
\end{gather*}
which exhibit $N$ and $M-N$ outgoing line solitons as $y\to \infty$ and $y\to -\infty$, respectively. This solution is uniquely determined by the ordered phase parameters $k_1 < \dots<k_m$ and a suitable coefficient matrix $A=(A_{n,m})$. 

In the following subsections, we elaborate on a subset of multi-soliton solutions that are traveling waves.  The following class of traveling wave solutions are those that we will modulate.

\subsection{Line soliton and its modulation}
\label{sec:line_soliton}

We can recover the line soliton solution \eqref{eq:line_soli} from the general form \eqref{eq:NSolitonSolution} with $(M,N)=(2,1)$, $k_1 < k_2$ and 
\begin{equation}
    \label{eq:LineSoliMatrix}
    A = \begin{pmatrix}
        1 & \alpha
    \end{pmatrix} , \quad \alpha > 0,
\end{equation}
where $\alpha$ is related to the $x$-intercept parameter as 
$$x_0 = \frac{\sqrt{12}}{k_1-k_2}\ln \alpha.$$
The amplitude and slope parameters of the line soliton can be expressed in terms of the phase parameters as follows
\begin{equation}\label{eq:LineSolitonParameters}
    a=\frac{1}{4}(k_--k_+)^2, \quad q=\frac{1}{2}(k_-+k_+) ,
\end{equation}
where we have identified $k_- = k_1$ and $k_+ = k_2$.  

We now provide an overview of modulation theory that has been developed for line solitons; see \cite{ryskamp_evolution_2021,ryskamp_modulation_2022}.  The two-parameter family of line soliton solutions \eqref{eq:line_soli} is modulated by allowing for slow variation of its two parameters in the direction transverse to its propagation.  Consequently, the modulated solution locally corresponds to a soliton with parameters $(k_-,k_+)$ that vary on the slow scales $Y = \epsilon y$, $T = \epsilon t$:  $k_\pm = k_\pm(Y,T)$ with $0 < \epsilon \ll 1$.  The line soliton modulation equations are \cite{ryskamp_evolution_2021}
\begin{subequations}
    \label{eq:mod_line_riemann}
    \begin{align}
        \label{eq:mod_line_k1}
        \partial_T k_\pm + \lambda_\pm \partial_Y k_\pm &= 0,
    \end{align}
    with characteristic velocities
    \begin{equation}\label{eq:lineSoliVelocity}
        \lambda_\pm(k_-,k_+)  = \frac{4}{3}k_\pm + \frac{2}{3}k_\mp .
    \end{equation}
\end{subequations}
Note that $k_\pm$ are Riemann invariants for the modulation system \eqref{eq:mod_line_riemann}.

With initial/boundary conditions provided for the modulation system \eqref{eq:mod_line_riemann}, its solution enables the construction of an approximate solution to KP \eqref{eq:KPII} as
\begin{gather}
    u(x,y,t) \sim  a(\epsilon y,\epsilon t) \operatorname{sech}^2\left(\sqrt{\frac{a(\epsilon y,\epsilon t)}{12}}\theta(x,y,t)\right), \\
    \theta(x,y,t) = x + \frac{1}{\epsilon}\int_0^{\epsilon y} q(Y',\epsilon t)\,\mathrm{d}Y' - \frac{1}{\epsilon} \int_0^{\epsilon t}(a(\epsilon y,T')/3 -q^2(\epsilon y,T'))\,\mathrm{d}T' .\label{eq:modulatedPhase}
\end{gather}
The parameter $1/\epsilon$ sets the large, hydrodynamic spatio-temporal scale for applicability of modulation theory.  Alternatively, one may set $\epsilon = 1$ and consider large $t$ and $y$ modulations, which we will do in the remainder of this paper.

Equation \eqref{eq:mod_line_riemann} admits two families of self-similar rarefaction wave (RW) solutions:  slow-RWs and fast-RWs.  For the slow-RWs, $k_+ = $ constant whereas for the fast-RWs, $k_- = $ constant. These are admissible weak solutions to the Riemann problem so long as the non-constant Riemann invariant is non-decreasing across the wave.

The slow-RW solution is 
\begin{subequations}
    \label{eq:1-RW}
    \begin{align}
         &k_-(y,t) = 
        \begin{cases}
            K^{l}_-& y-y_0 < \lambda_-(K_-^l,K_+) t, \\
            \frac{3(y-y_0)}{4t} - \frac{1}{2}K_+ & \lambda_-(K_-^{l}, K_+) t < y-y_0 < \lambda_-(K_-^{u}, K_+) t, \\
            K^u_- & y-y_0>\lambda_-(K_-^u,K_+)t, \\ 
        \end{cases}\\
        &k_+(y,t) \equiv K_+,
    \end{align}
\end{subequations}
which is admissible so long as $K_-^l < K_-^u < K_+$.  Translating admissibility to the physical parameters $(a,q)$ in eq.~\eqref{eq:LineSolitonParameters}, a slow-RW represents a continuous transition between the constant lower parameter values $(a^l,q^l)$ and constant upper values $(a^u,q^u)$ subject to the conditions
\begin{equation}\label{1-RW}
        a^l > a^u,\quad\sqrt{a^l} - \sqrt{a^u} = q^u -q^l.
\end{equation} 

The fast-RW solution is
\begin{subequations}
    \label{eq:2-RW}
    \begin{align}
         &k_- \equiv K_-,\\
         &k_+(y,t) = 
        \begin{cases}
            K_+^l & y-y_0 < \lambda_+(K_-,K_+^l) t, \\
            \frac{3(y-y_0)}{4t} - \frac{1}{2}K_- & \lambda_+(K_-,K_+^l) t < y-y_0 < \lambda_+(K_-,K_+^u) t, \\
            K_+^u & y-y_0> \lambda_+(K_-,K_+^u)t ,
        \end{cases}
    \end{align}
\end{subequations}
which is admissible when $K_- < K_+^l < K_+^u$.  In terms of the physical parameters $(a,q)$ in eq.~\eqref{eq:LineSolitonParameters}, a fast-RW is a continuous transition between the constant lower parameter values $(a^l,q^l)$ and constant upper values $(a^u,q^u)$ subject to the conditions
\begin{equation}\label{2-RW}
        a^l < a^u, \quad \sqrt{a^u} - \sqrt{a^l} = q^u -q^l.
\end{equation}

\subsection{Y-soliton}
\label{sec:y_soliton}

The Y-soliton solution corresponds to either $(M,N)=(3,1)$ or $(M,N)=(3,2)$ depending on whether it has one outgoing soliton as $y \to \infty$ or $y \to -\infty$, respectively. This is a three parameter family of solutions determined by $k_1 < k_2 < k_3$ and 
\begin{equation}
    \label{eq:y_soliton_A_matrix}
    A = \begin{pmatrix}
        1 & 0 & -\beta \\
        0 & 1 & \alpha 
    \end{pmatrix}, \quad \alpha, \beta > 0,
\end{equation}
for the $N = 2$ case.  
We set $\alpha = \frac{k_2-k_1}{k_3-k_2}$, $\beta = \frac{k_2-k_1}{k_3-k_1}$ so that the Y-soliton vertex at $t=0$ is located at the origin $(0, 0)$.
\begin{figure}[h!]
    \centering
    \includegraphics{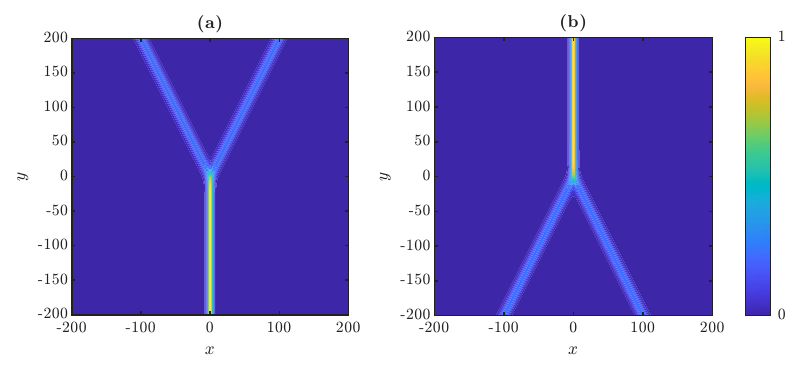}
    \caption{In this figure, \textbf{(a)} shows the $M=3, N=2$ case while \textbf{(b)} shows the $M=3, N=1$ case. In both cases, $(k_1,k_2,k_3) = (-1, 0,1)$.}
    \label{fig:twoYs}
\end{figure}

For this paper, we will focus our analysis on the $(M,N)=(3,2)$ Y-soliton that is ``upright''. Note that our analysis extends to the ``inverted'' Y-soliton, the $(M,N)=(3,1)$ case by the $y\to-y$ invariance of the KP equation \eqref{eq:KPII}.

As $y \to -\infty$ the Y-soliton solution asymptotes to the line soliton with parameters $(k_1,k_3)$ so we label this outgoing soliton by its non-dominant parameter as leg $j = 2$.  The other two asymptotic soliton legs as $y \to \infty$ are labeled similarly; see Fig.~\ref{fig:YSoliLabeling}.  This amounts to a counterclockwise labeling convention in which the upper leftmost outgoing soliton is indexed as $j = 1$.  We label the physical parameters for each leg by the leg index.  For example, $(a^{(1)},q^{(1)})$ correspond to the amplitude and slope of leg $j = 1$ so that eq.~\eqref{eq:LineSolitonParameters} implies $a^{(1)} = (k_3-k_2)^2/4$ and $q^{(1)} = (k_2+k_3)/2$.
\begin{figure}[h!]
    \centering
    \includegraphics{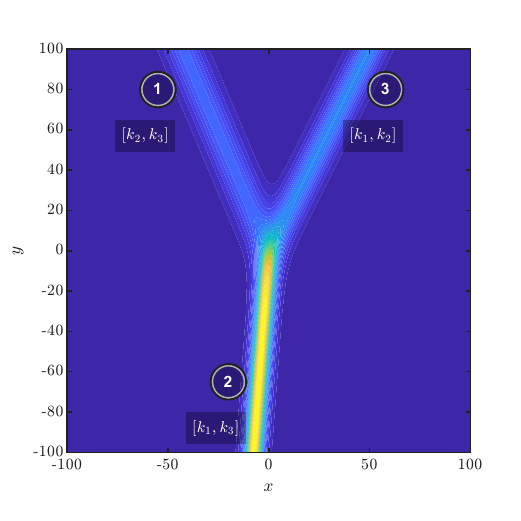}
    \caption{The circled number denotes leg $j$, while the $[k_m, k_n]$ represents the two dominant leg parameters on that leg.}
    \label{fig:YSoliLabeling}
\end{figure}

The velocity $\mathbf{s}$ of the Y-soliton vertex where the soliton legs meet is obtained by setting $\theta_1 = \theta_2 = \theta_3$, which yields
\begin{equation} \label{eq:ySoliVel}
    \mathbf{s} = (s_x,s_y) = \left(-\frac{1}{3}(k_1k_2+k_1k_3+k_2k_2),\frac{2}{3}(k_1+k_2+k_3) \right).
\end{equation}
This confluence of three partial solitons at the vertex have led researchers to characterize the Y-soliton as \textit{resonant} \cite{miles_resonantly_1977}.  This resonance arises in the Mach reflection problem in which a slightly oblique soliton reflects off of a solid boundary.  The two soliton legs as $y \to \infty$ correspond to the incident and reflected solitons.  The soliton leg as $y \to -\infty$ has largest amplitude $a^{(2)} = \tfrac{1}{4}(k_3 - k_1)^2$ and is called the ``Mach stem'', which is perpendicular to the boundary.

\subsection{X-soliton}
\label{sec:x_soliton}

There are two types of X-soliton solutions that are traveling waves: the O-type and P-type, also known the 2143 soliton and 4321, respectively \cite{kodama_kp_2010}.  They both correspond to $(M,N)=(4,2)$ in the tau function formulation \eqref{eq:NSolitonSolution} with different coefficient matrices.  The two outgoing line solitons as $y\to \infty$  consist of the parameters $k_1 < k_2 < k_3 < k_4$.  The two outgoing line solitons as $y\to -\infty$ have the same parameters so that these represent elastic 2-soliton interactions.  The two cases are qualitatively distinguished by the nature of the interaction.  For the X-soliton solutions, we will take a more direct approach to labeling the legs with a tuple corresponding to the two dominant phase parameters on each leg. 

The O-type soliton has the coefficient matrix
\begin{equation}\label{eq:OTypeMatrix}
    A = \begin{pmatrix}
        1 & \alpha & 0 &0\\
        0 & 0 & 1 & \beta
    \end{pmatrix} \quad \alpha,\beta >0,
\end{equation}
where we set $\alpha = \frac{k_3-k_1}{k_3-k_2}$, $\beta= \frac{k_3-k_1}{k_4-k_1}$ so that the vertex is located at $(0,0)$. The vertex velocity is
$\mathbf{s} = (s_x,s_y)$, where its components
\begin{align}
s_x &=
\frac{
    (k_1+k_2)\bigl(k_4^2+k_3k_4+k_3^2\bigr)
    -
    (k_3+k_4)\bigl(k_2^2+k_1k_2+k_1^2\bigr)
}{
    3\,(k_1+k_2-k_3-k_4)
},
\\
s_y &=
\frac{2}{3}\left[ k_1 + k_2 + k_3 + k_4
-
\frac{k_1 k_2 - k_3 k_4}
     {k_1+k_2-k_3-k_4} \right],\label{eq:O_type_yvel}
\end{align}
are obtained by the intersection of the phase-balance lines $\theta_1-\theta_2 = \text{const}$ and   $\theta_3-\theta_4 = \text{const}$.  An example O-type soliton is depicted in Fig.~\ref{fig:o_type_soliton}.  Qualitatively, the O-type soliton represents the approximate superposition of two oblique line solitons experiencing a ``head-on'' collision, so represents a somewhat weaker nonlinear interaction than the P-type soliton discussed next.  In the classification scheme of ref \cite{biondini_line_2007}, the O-type soliton is an ``ordinary'' two-soliton solution. In contrast to the Mach reflection associated with the Y-soliton, the O-type X-soliton soliton models the regular reflection of an oblique soliton off of a solid boundary \cite{miles_obliquely_1977}.  For this, only the two legs of the upper half of the O-type soliton should be considered where the rightmost leg is incident upon and the leftmost leg is reflected from the boundary.

\begin{figure}[h!]
    \centering
    \includegraphics{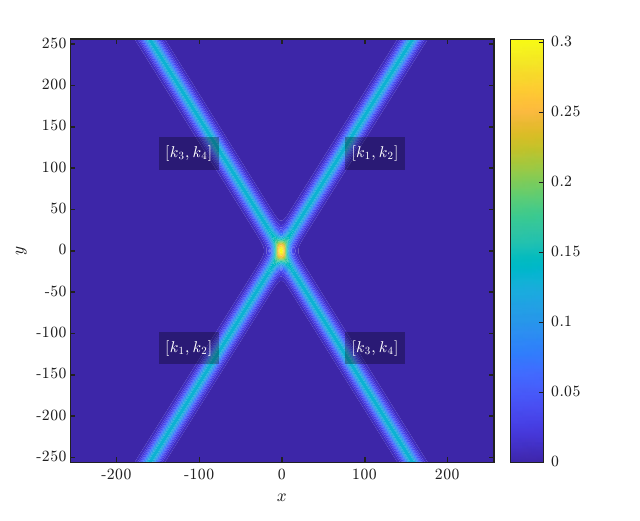}
    \caption{Example O-type soliton exhibiting a ``head-on'' collision between two line solitons.  The maximum of the solution is greater than the sum of the incoming soliton amplitudes.}
    \label{fig:o_type_soliton}
\end{figure}

The P-type soliton corresponds to the coefficient matrix
\begin{equation}\label{eq:PTypeMatrix}
    A = \begin{pmatrix}
        1 & 0 & 0 &-\beta\\
        0 & 1 & \alpha & 0
    \end{pmatrix} \quad \alpha,\beta >0
\end{equation}
where we set $\alpha = \frac{k_2-k_1}{k_3-k_1}$, $\beta = \frac{k_2-k_1}{k_4-k_2}$ so that the vertex is located at $(0,0)$. The vertex velocity $\mathbf{s} = (s_x,s_y)$ has components
\begin{align}
s_x &=
\frac{
  (k_1+k_4)\bigl(k_2^2 + k_2 k_3 + k_3^2\bigr)
  -
  (k_2+k_3)\bigl(k_1^2 + k_1 k_4 + k_4^2\bigr)
}{
  3\bigl(k_1+k_4-k_2-k_3\bigr)
},
\\
s_y &=
\frac{2}{3}
\left[
k_1+k_2+k_3+k_4
-\frac{k_1k_4-k_2k_3}{k_1+k_4-k_2-k_3}
\right],\label{eq:P_type_yvel}
\end{align}
that are obtained by intersecting the phase-balance lines $\theta_2-\theta_3 = \text{const}$ and   $\theta_1-\theta_4 = \text{const}$. An example P-type soliton is shown in Fig.~\ref{fig:p_type_soliton}.  Qualitatively, the P-type soliton represents a stronger nonlinear ``overtaking'' interaction between two line solitons \cite{miles_obliquely_1977}.  In \cite{biondini_line_2007}, the P-type soliton is classified as an ``asymmetric'' two-soliton interaction.  
\begin{figure}[h!]
    \centering
    \includegraphics{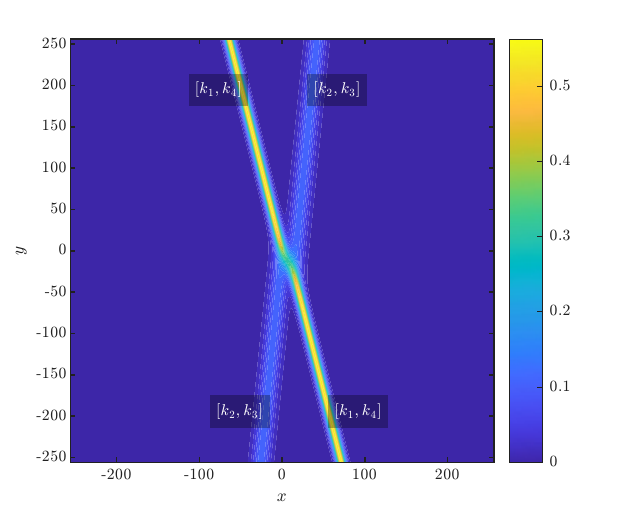} %
    \caption{Example P-type solitons exhibiting an ``overtaking'' interaction between two line solitons.  The interaction region exhibits amplitudes that are smaller than the sum of the incoming soliton amplitudes.}
    \label{fig:p_type_soliton}
\end{figure}

\section{Modulation equations of Y-soliton traveling wave solution}
\label{sec:modul-y-solit}

In order to describe modulations of Y-soliton solutions, we will first consider its zero dispersion limit.  This limit is obtained by the hydrodynamic rescaling $(x,y,t) \to (\epsilon x, \epsilon y, \epsilon t)$ of the KPII equation \eqref{eq:KPII} and therefore its multi-soliton solutions \eqref{eq:NSolitonSolution}.  The parameter $\epsilon > 0$ is proportional to the width of each soliton leg.  Consequently, soliton leg $j$ of a Y-soliton solution becomes localized along a line $x + q^{(j)} y - c^{(j)} t = $const in the zero dispersion limit $\epsilon \to 0$ (also known as the tropical limit) that results in a soliton graph \cite{kodama_kp_2010}. 
The three asymptotic line solitons that compose the Y-soliton are localized to three semi-infinite lines with a point of intersection at a vertex. At the vertex, the Y-soliton has been shown to exhibit a discontinuity in soliton leg parameters that satisfies certain Rankine-Hugoniot conditions \cite{ryskamp_modulation_2022}.  This representation of a Y-soliton motivates the treatment of the discontinuity as a boundary between each soliton leg.  Thus, Y-soliton modulation consists of line soliton modulation of each leg and appropriate boundary conditions at the Y-soliton vertex.

On leg $j$, two phase parameters $k_m^{(j)}$ and $k_n^{(j)}$, $m < n$ are dominant (cf.~Fig.~\ref{fig:YSoliLabeling}).  Because we treat leg $j$ as a semi-infinite line soliton, the modulated phase parameters satisfy eq.~\eqref{eq:mod_line_riemann}
\begin{subequations}
    \label{eq:mod_Y_riemann}
    \begin{gather}
        \label{eq:kDef}
        \frac{\partial k_{m,n}^{(j)}}{\partial t} + \lambda^{(j)}_{m,n} \frac{\partial k_{m,n}^{(j)}}{\partial x} = 0, \\
        \label{eq:velDef}
        \lambda^{(j)}_{m,n} = \frac{4}{3}k^{(j)}_{m,n} + \frac{2}{3}k^{(j)}_{n,m}, \quad 1\leq n<m\leq3,\quad n,m\neq j, \quad j\in \{1,2,3 \},
    \end{gather}
\end{subequations}
subject to a domain restriction defined by the leg positioning relative to the moving vertex.  It will be helpful to notate the characteristic velocity by an explicit ordering.  Since $n < m$, we will always have $\lambda_n^{(j)} < \lambda_m^{(j)}$ for $n,m \ne j$. So we will write $\lambda_-^{(j)} = \lambda_n^{(j)}$ for the slow characteristic velocity and $\lambda_+^{(j)} = \lambda_m^{(j)}$ for the fast characteristic velocity on leg $j$, with $n$ and $m$ implied. Now we impose boundary conditions.  Specifically, the boundary condition between legs $j=1$ and $j=2$ is continuity of the parameter $k_3$ at the vertex $y = p(t)$, legs $j=1$ and $j=3$ exhibit continuity of $k_2$, and between legs $j=2$ and $j=3$, the parameter $k_1$ is continuous. In short,
\begin{equation}\label{eq:boundaryCont}
        k^{(i)}_m(p(t), t) = k^{(j)}_m(p(t), t), \quad i,j,m\in \{1,2,3 \}, \quad m\notin \{i,j\} .
\end{equation}
Figure \ref{fig:eqLabeling} highlights the relevant shared phase parameter that is continuous between each pair of legs in the Y-soliton.  The vertex position $p(t)$ satisfies the ODE
\begin{equation}\label{eq:ODEDefOfP}
    \dot{p}(t) =\frac{2}{3}(k_1+k_2+k_3) |_{y=p(t)}, \quad p(0) = 0,
\end{equation}
where we need not distinguish the leg number for each of the phase parameters because of continuity \eqref{eq:boundaryCont}.  We also assume that the vertex is initially located at the origin $y = 0$.  
We thus write the three pairs of equations \eqref{eq:mod_Y_riemann} that govern modulations along the Y-soliton as
\begin{subequations}
    \label{eq:mod_param_k3}
    \begin{align}
        \label{eq:mod_param_k3_leg2}
        \partial_t k^{(2)}_3 + \lambda^{(2)}_+ \partial_y k^{(2)}_3 &= 0 \quad y<p(t),
        \\
        \label{eq:mod_param_k3_leg1}
        \partial_t k^{(1)}_3 + \lambda^{(1)}_+ \partial_y k^{(1)}_3 &= 0 \quad y>p(t),
    \end{align}
\end{subequations}

\begin{subequations}
    \label{eq:mod_param_k2}
    \begin{align}
        \label{eq:mod_param_k2_leg3}
        \partial_t k^{(3)}_2 + \lambda^{(3)}_+ \partial_y k^{(3)}_2 &= 0 \quad y>p(t),
        \\
        \label{eq:mod_param_k2_leg1}
        \partial_t k^{(1)}_2 + \lambda^{(1)}_- \partial_y k^{(1)}_2 &= 0 \quad y>p(t),
    \end{align}
\end{subequations}

\begin{subequations}
    \label{eq:mod_param_k1}
    \begin{align}
        \label{eq:mod_param_k1_leg3}
        \partial_t k^{(3)}_1 + \lambda^{(3)}_- \partial_y k^{(3)}_1 &= 0 \quad y>p(t),
        \\
        \label{eq:mod_param_k1_leg2}
        \partial_t k^{(2)}_1 + \lambda^{(2)}_- \partial_y k^{(2)}_1 &= 0 \quad y<p(t).
    \end{align}
\end{subequations}
Overall, eqs.~\eqref{eq:mod_param_k3}--\eqref{eq:mod_param_k1} are a set of hyperbolic equations for the six quantities $k^{(j)}_m$, where $j,m \in \{1,2,3 \}$, $j\neq m$ subject to continuity \eqref{eq:boundaryCont} at the boundary \eqref{eq:ODEDefOfP}.  To formulate a well-posed problem, we must supplement this system with initial conditions $k_m^{(j)}(y,0)$ on the corresponding domain, either $y > 0$ or $y < 0$.

\begin{figure}[h!]
    \centering
    \includegraphics{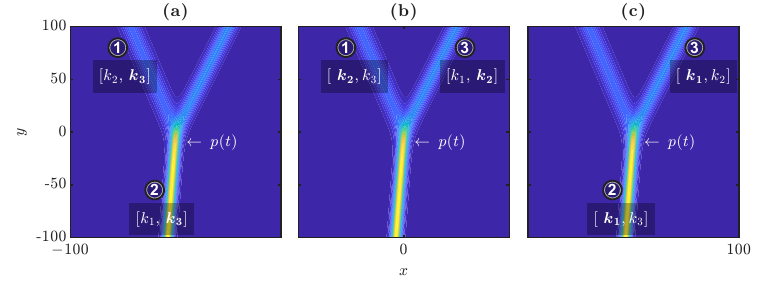}
    \caption{Depiction of Y-soliton modulation boundary conditions at the vertex $p(t)$.  \textbf{(a)} Continuity of $k_3$ between legs 1 and 2 corresponding to system \eqref{eq:mod_param_k3}. \textbf{(b)} Continuity of $k_2$ between legs 1 and 3 corresponding to system \eqref{eq:mod_param_k2}. \textbf{(c)} Continuity of $k_1$ corresponding to system \eqref{eq:mod_param_k1}.}
    \label{fig:eqLabeling}
\end{figure}

\subsection{Linearized modulation equations}
    We will start by analyzing the linearization of the system \eqref{eq:mod_param_k3}-\eqref{eq:mod_param_k1} in which each $k$ parameter is linearized by writing $k^{(j)}_m= K_m + \epsilon \bar{k}^{(j)}_m(y,t)$, where 
    $K_1<K_2<K_3$
    are constants and $0 < \epsilon \ll 1$ quantifies the small deviation from the constant value. 
    The modulation system reduces to a system of linear transport equations in the form 
    \begin{subequations}\label{eq:linearized_form_split}
        \begin{align}
            \partial_t \bar{k}^{(i)}_m + \bar{\lambda}_\pm^{(i)} \partial_y \bar{k}^{(i)}_m &= 0 \quad y>\bar{p}(t)\\
            \partial_t \bar{k}^{(j)}_m + \bar{\lambda}_\pm^{(j)} \partial_y \bar{k}^{(j)}_m &= 0 \quad y<\bar{p}(t)
        \end{align}
    \end{subequations}
    for the split domain cases $m = 1$ or $m = 3$ from equations \eqref{eq:mod_param_k3} and \eqref{eq:mod_param_k1}.  The remaining equation \eqref{eq:mod_param_k2} is written
    \begin{subequations}\label{eq:linearized_form_same}
        \begin{align}
            \partial_t \bar{k}^{(3)}_2 + \bar{\lambda}_+^{(3)} \partial_y \bar{k}^{(3)}_2 &= 0 \quad y>\bar{p}(t)\\
            \partial_t \bar{k}^{(1)}_2 + \bar{\lambda}_-^{(1)} \partial_y \bar{k}^{(1)}_2 &= 0 \quad y>\bar{p}(t) .
        \end{align}
    \end{subequations}
    The characteristic velocities in \eqref{eq:linearized_form_split} and \eqref{eq:linearized_form_same} are
    \begin{align}\label{eq:linearizedVelDef}
        \bar{\lambda}^{(j)}_+ = \frac{4}{3}K_m + \frac{2}{3}K_n, \quad \bar{\lambda}^{(j)}_- = \frac{4}{3}K_n + \frac{2}{3}K_m,
    \end{align}
    where $1\leq n<m\leq3$, $n,m\neq j$, $j\in \{1,2,3 \}$.  Linearizing \eqref{eq:ODEDefOfP} implies the boundary moves with constant velocity so we can recover its motion in advance to be
    \begin{equation} \label{eq:linearized_position}
       \bar{p}(t) = \frac{2}{3}(K_1+K_2+K_3)t .
    \end{equation}

    Equations \eqref{eq:linearized_form_split}--\eqref{eq:linearized_position} are accompanied by the initial conditions
    \begin{subequations}\label{eq:linearizedIC}
        \begin{align}
            & \bar{k}^{(i)}_m(y, 0) = \bar{k}^{(i)}_{m,0}(y), \quad i,m\in \{1,2,3 \}, \quad i \ne m,
        \end{align}
    \end{subequations}
    and continuity at the boundary
    \begin{equation} \label{eq:linearizedBC}
        \bar{k}^{(i)}_m(\bar{p}(t), t) = \bar{k}^{(j)}_m(\bar{p}(t), t), \quad i,j,m\in \{1,2,3 \}, \quad m\notin \{i,j\} .
    \end{equation}

    The initial-boundary value problem \eqref{eq:linearized_form_split}--\eqref{eq:linearizedBC} admits a unique solution. We state this as a theorem.
    \begin{theorem}\label{th:linearizedWellPosed}
        Assume the initial data for the linearized modulation system belong to $C^1$ and that they are compatible at the boundary. That is
        \[
        \bar{k}^{(1)}_{3,0}(y=0)=\bar{k}^{(2)}_{3,0}(y=0),\qquad
        \bar{k}^{(1)}_{2,0}(y=0)=\bar{k}^{(3)}_{2,0}(y=0),\qquad
        \bar{k}^{(2)}_{1,0}(y=0)=\bar{k}^{(3)}_{1,0}(y=0),
        \]
        and similar for the first spatial derivatives of $\bar{k}_{m,0}^{(i)}$.
        Then the initial-boundary value problem
        \eqref{eq:linearized_form_split}--\eqref{eq:linearized_form_same},
        subject to \eqref{eq:linearizedIC} and \eqref{eq:linearizedBC}, is well-posed.
    \end{theorem}

    \begin{proof}
    The proof of this follows by counting the number of characteristics into and out of the boundary, then constructing the exact solution.

    Consider the co-moving frame        
        \[
        c:=\frac23(K_1+K_2+K_3),\qquad \xi=y-\bar p(t)=y-ct.
        \]
        In the transformed $(\xi,t)$ variables, each equation is a linear transport equation
        \[
        \partial_t u + v_\pm^{(j)}\,\partial_\xi u=0,
        \qquad v_\pm^{(j)}=\bar\lambda_\pm^{(j)}-c,
        \]
        with constant coefficient $v$ on a fixed domain, either $\xi > 0$ or $\xi < 0$.
        
        The ordering $K_1<K_2<K_3$ implies that the six relative speeds have the following signs
\begin{subequations}\label{eq:linearizedCharSpeeds}
        \begin{align}
        \begin{split}
        v_+^{(1)} &= \frac{2}{3}(K_3 - K_1) > 0, \qquad
        v_-^{(1)} = \frac{2}{3}(K_2 - K_1) > 0
        \end{split}
        \label{eq:linearizedCharSpeeds_1}\\
        \begin{split}
        v_+^{(2)} &= \frac{2}{3}(K_3 - K_2) > 0, \qquad
        v_-^{(2)} = -\frac{2}{3}(K_2 - K_1) < 0
        \end{split}
        \label{eq:linearizedCharSpeeds_2}\\
        \begin{split}
        v_+^{(3)} &= -\frac{2}{3}(K_2 - K_3) < 0, \qquad
        v_-^{(3)} = -\frac{2}{3}(K_3 - K_1) < 0
        \end{split}
        \label{eq:linearizedCharSpeeds_3}
        \end{align}
        \end{subequations}
        For \(\xi<0\), a positive relative speed \(a>0\) carries information towards the boundary, whereas a negative
        relative speed \(a<0\) carries information away.  When \(\xi>0\), the orientation is reversed.
        On $\xi<0$, the speeds in \eqref{eq:linearizedCharSpeeds_2} imply that the only data flowing into the boundary is $\bar k_3^{(2)}$. On $\xi>0$, \eqref{eq:linearizedCharSpeeds_1} and \eqref{eq:linearizedCharSpeeds_3} imply that both $\bar k_1^{(3)}$ and $\bar k_2^{(3)}$ flow into the boundary. The remaining three parameters are transported away from the boundary.  Succinctly stated, eq.~\eqref{eq:linearizedCharSpeeds} implies that three characteristic families enter the boundary and three exit.
        
        Using the continuity conditions at the boundary \eqref{eq:linearizedBC},
        along with the initial conditions
        determine the solution for all $\xi$.
        The solution is obtained by the method of characteristics with its uniqueness a standard result from ODE theory.  The calculation is standard.  For example, the solution for $\bar k_2^{(1,3)}$ is

        \begin{subequations}\label{eq:linearizedSolTok2}
            \begin{equation}
                 \bar{k}_2^{(3)}(\xi, t) =\bar{k}_{2,0}^{(3)}(\xi-v^{(3)}_{+} t) \quad \text{for } \xi>0, t>0,
            \end{equation}
            \begin{equation}
                    \bar{k}_2^{(1)}(\xi, t) =\begin{cases}
                    \bar{k}_{2,0}^{(1)}(\xi-v^{(1)}_{-} t)   & \text{for } \xi-v^{(1)}_{-} t > 0, \\ 
                    \bar{k}_{2,0}^{(3)}\left( 
                        \frac{v^{(3)}_{+}}
                        {v^{(1)}_{-}}
                        (\xi-v^{(1)}_{-} t) 
                    \right) & \text{for } \xi-v^{(1)}_{-} t \leq 0.
                    \end{cases} 
            \end{equation}
        \end{subequations}

        Continuous dependence on the initial data follows immediately from the explicit solution, thus concluding our proof of well-posedness.
        \end{proof}

  We note that the $C^1$ regularity assumption on the initial data could be relaxed.  In such a scenario, the solution, e.g., \eqref{eq:linearizedSolTok2}, is the weak solution.
    
    An interesting implication of Theorem \ref{th:linearizedWellPosed} is the dynamical, linearized stability of Y-solitons subject to long wavelength, transverse perturbations due to the boundedness of the linearized solution. The precise spectral stability of the Y-soliton solution of KPII has recently been studied in \cite{Mizumachi2026}.

\subsection{Local well-posedness for the nonlinear modulation system}
\label{sec:well_posedness_nonlinear}

A corollary of Theorem \ref{th:linearizedWellPosed} is the local well-posedness for the nonlinear modulation system in the vicinity of the boundary.

\begin{corollary}
\label{th:nonlinear_well_posed}
If the dependent variables are continuous and strictly ordered 
\begin{equation}\label{eq:localVertexOrdering}
    \kappa_{1}(t)<\kappa_{2}(t)<\kappa_{3}(t),
\end{equation}
for $t$ and $y$
in open intervals, the latter containing $y = p(t)$ satisfying \eqref{eq:ODEDefOfP} where
\[
    \kappa_{3}:=k_{3}^{(1)}(p(t),t)=k_{3}^{(2)}(p(t),t),
\]
\[
    \kappa_{2}:=k_{2}^{(1)}(p(t),t)=k_{2}^{(3)}(p(t),t),
\]
and
\[
    \kappa_{1}:=k_{1}^{(2)}(p(t),t)=k_{1}^{(3)}(p(t),t),
\]
then the nonlinear modulation system
\eqref{eq:mod_param_k3}--\eqref{eq:mod_param_k1} subject to
\eqref{eq:boundaryCont} and
\eqref{eq:ODEDefOfP}
remains well-posed within the same intervals. 
\end{corollary}

\begin{proof}

The assumptions imply that $p(t)$ is continuously differentiable, which allows us to introduce the moving coordinate frame $\xi=y-p(t)$,
to fix the vertex position at $\xi = 0$. This allows us to follow the argument from the linear case, transforming to equations that take the form 
\[
        \partial_t u + v_\pm^{(j)}\,\partial_\xi u=0,
        \qquad v_\pm^{(j)}=\lambda_\pm^{(j)}(t)-\dot{p}(t),
\]
where the ordering \eqref{eq:localVertexOrdering}, allows for analogs of relations \eqref{eq:linearizedCharSpeeds} to hold. Then, exactly three characteristics enter and three exit the vertex implying the local well-posedness.

\end{proof}

There are a few important notes and caveats for Corollary \ref{th:nonlinear_well_posed} that we discuss below.
Although we proved local stability in a neighborhood of the vertex, the global behavior may be different. In fact, we will introduce a case where increasing Riemann data on one leg maintains the local stability around the vertex but leads to gradient catastrophe away from the vertex on another leg. 
This will be regularized by introducing a new vertex, corresponding to a shock in the modulation system and results in a global topological change in the soliton.  However, locally, around the vertex, the modulation system is still well-posed, and accordingly still resembles a Y-soliton.

On the other hand, we also present initial conditions that violate a generalization of the ordering condition \eqref{eq:localVertexOrdering} on some leg, sufficiently far away from the vertex.  We will show that this data evolves in such a way that the vertex ordering condition \eqref{eq:localVertexOrdering} is maintained for all time. This allows for the application of Corollary \ref{th:nonlinear_well_posed}, preserving the well-posedness at the vertex and the local structure of the Y-soliton for finite times.

Lastly, we will also discuss one case that violates the ordering condition \eqref{eq:localVertexOrdering} at the vertex and leads to a genuine topological change of the Y-soliton around the vertex.

\section{Classification of RW solutions to the Y-soliton modulation system}
\label{sec:ySoliModulations}
    We will classify the rarefaction wave solutions for the $N=2$ Y-soliton.
    That is, for a non-modulated Y-soliton with parameters $(K_1, K_2, K_3)$, we choose a leg index $j \in \{1,2,3\}$ and a parameter index $m\in \{1, 2, 3\}$, $m \ne j$ and prescribe a step in $k_m$ from $K_m$ to $K'_m$ located some distance from the Y-soliton vertex.  The initial data must be admissible, resulting in the generation of a rarefaction wave, which will necessarily order $K_m$ and $K_m'$.  We will classify all possible such modulation solutions resulting from these initial conditions and comment on their significance. We avoid initial conditions that result in the immediate growth of a new leg as studied in \cite{ryskamp_modulation_2022, han_regularizations_2026}.  In what follows, emphasis is given to those modulations that interact with the vertex.

    \subsection{Leg $j = 1$}

    On leg $j = 1$, admissible rarefaction data consists of either $K_2 < K_2' < K_3$ ($m = 2$) or $K_3 < K_3'$ ($m = 3$).
    For both cases, the rarefaction wave propagates away from the vertex. This can be seen from the characteristic velocity ordering described in the proof of Theorem \ref{th:linearizedWellPosed}.  In particular, $\lambda^{(1)}_\pm(K_2,K_3) = \tfrac{4}{3}K_{3,2} + \tfrac{2}{3}K_{2,3} > \dot{p} = \tfrac23(K_1+K_2+K_3)$.  Thus, both rarefaction wave modulations that originate on leg $j = 1$, remain on leg $j = 1$ for all $t > 0$ so that they represent single-leg modulations.

    \subsection{Leg $j = 2$}

    For leg $j = 2$, the slow-RW with $m = 1$ and $K_1' < K_1$ propagates away from the vertex because $\lambda_-^{(2)}(K_1,K_3) = \tfrac{4}{3}K_1 + \tfrac{2}{3}K_3 < \dot{p} = \tfrac23(K_1+K_2+K_3)$, so it is just a single leg modulation remaining on leg $j = 2$.
    
The fast-RW 
\begin{equation}
   m = 3:\quad K_1 < K_3' < K_3,
\end{equation}
initiated on leg $j =2$ is our first example of a modulation that will interact with the vertex. For this case, the initial condition to the system \eqref{eq:mod_param_k3}--\eqref{eq:mod_param_k1} is
    \begin{equation}
        \begin{aligned}
            &k^{(2)}_3(y,0)=  \begin{cases}
            K_3'& y<y_0, \\
            K_3& y_0 < y< 0, \\ 
            \end{cases}\\
            &k^{(1)}_3(y,0) = K_3, \quad y > 0,\\
            &k^{(j)}_2(y,0) = K_2, \quad y > 0, \quad j\in \{1,3 \},\\
            &k^{(j)}_1(y,0) = K_1, \quad (-1)^j y < 0, \quad j\in \{2,3 \}.
        \end{aligned}
        \label{eq:fastRWj2IC}
    \end{equation}
    The only variable that is modulated is $k_3^{(1,2)}$, so the primary equation of interest is \eqref{eq:mod_param_k3}.  See Fig.~\ref{fig:fastRWj2IC} for a labeled initial configuration.

    \begin{figure}
        \centering
        \includegraphics{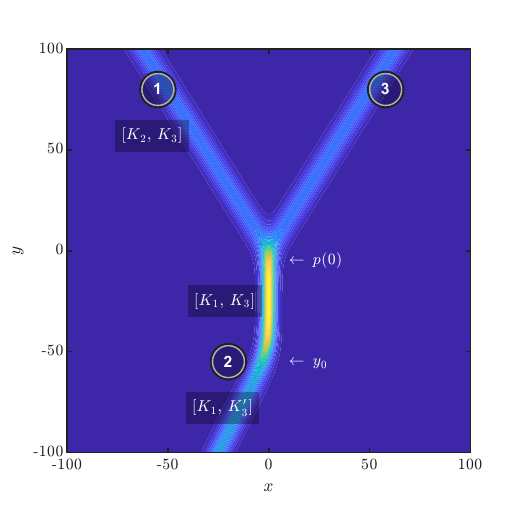}
        \caption{Sketch of the initial conditions \eqref{eq:fastRWj2IC} for the fast-RW initiated on leg $j = 2$.  The ensuing modulation is confined to legs $j = 2$ and $j = 1$.}
        \label{fig:fastRWj2IC}
    \end{figure}

    \subsubsection{Case $K_1 < K_2 < K_3' < K_3$}
    
    In this subsection, we will restrict to the ordering 
    \begin{equation}\label{eq:fastRWj2Ordering}
        K_1<K_2<K_3'<K_3 .
    \end{equation}
    The case $K_1 < K_3' \le K_2$ will be considered next.
    Solving for the initial conditions \eqref{eq:fastRWj2IC}, on the $j=2$ leg prior to the fast-RW reaching the vertex, the solution for $y < p(t)$, $t > 0$ is (cf.~eq.~\eqref{eq:2-RW})
    \begin{equation} \label{eq:2RWSolnj2Onj2}
        k^{(2)}_3(y, t) = \begin{cases}
        K_3' & \text{if } y-y_0 < \lambda^{(2)}_+(K_1, K_3') t, \\
        \frac{3(y-y_0)}{4t} - \frac{1}{2}K_1 & \text{if } \lambda^{(2)}_+(K_1, K_3') t < y-y_0 < \lambda^{(2)}_+(K_1, K_3) t, \\
        K_3 & \text{if } y-y_0 > \lambda^{(2)}_+(K_1, K_3) t .
        \end{cases}
    \end{equation}
    This solution provides sufficient information to solve the vertex position ODE \eqref{eq:ODEDefOfP} at all times since all other phase parameters are constant. Thus, the vertex position is
    \begin{equation}\label{eq:2RWj2Vertex}
        p(t) = \begin{cases}
        \frac{2}{3} \left(K_1 + K_2 + K_3 \right)t, & t < t_1, \\
        \left(1-2\sqrt{\frac{2(K_3 -K_2)}{3|y_0|}t}\right)y_0 + \left(\frac{2}{3}K_1  + \frac{4}{3}K_2\right) t,  & t_1 <t < t_2,\\
         \frac{2}{3}  (K_1+K_2 +K_3' )t + y_2,& t > t_2,
        \end{cases}
    \end{equation}
    where we introduce the RW entry time to the vertex  $t = t_1$
    \begin{equation} \label{FastRWj2t1}
        t_1=\frac{3|y_0|}{2(K_3 -K_2)},
    \end{equation}
    and exit time from the vertex  $t = t_2$
    \begin{equation}\label{FastRWj2t2}
        t_2= \left ( \frac{K_3-K_2}{K_3'-K_2} \right )^2 t_1 , 
    \end{equation}
    which satisfy $0<t_1<t_2$ by \eqref{eq:fastRWj2Ordering} and $y_0<0$.  The accumulated position phase of the vertex after RW interaction is
    \begin{equation}
       y_2 = \frac{K_3 - K_3'}{K_3' - K_2} |y_0|.
    \end{equation}
    Knowing the position of the boundary and invoking continuity at the boundary  \eqref{eq:boundaryCont}, we can solve for $k^{(1)}_3(y,t)$ with the method of characteristics for $y > p(t)$, $t > t_1$
    \begin{subequations}
        \begin{equation}
            k^{(1)}_3(y,t)= \begin{cases}
            K_3 & y \geq (\frac{4}{3}K_3 + \frac{2}{3}K_2 )t, \\
            f(\hat{t}^{(3)}(y,t)) &p(t) < y < (\frac{4}{3}K_3 + \frac{2}{3}K_2 )t,
            \end{cases}
        \end{equation}
        where $f(\hat{t}^{(3)})$ results from tracing back along a fast characteristic to the boundary 
        \[
        f(\hat{t}^{(3)})  = \begin{cases}
        \frac{3(p(\hat{t}^{(3)})-y_0)}{4\hat{t}^{(3)}} - \frac{1}{2}K_1 & \text{if } t_1 <\hat{t}^{(3)}< t_2\\
        K_3' & \text{if } \hat{t}^{(3)} > t_2
        \end{cases} 
        \]
        and $\hat{t}^{(3)}=\hat{t}^{(3)}(y,t)$ is the foot of the characteristic on leg $j=3$, defined implicitly as
        \[
        y  = (\frac{4}{3}f(\hat{t}^{(3)}) + \frac{2}{3}K_2)(t-\hat{t}^{(3)}) +p(\hat{t}^{(3)}) .
        \]
    \end{subequations}
    The solution for $0 < t < t_1$, $y > p(t)$ on leg $j = 1$ is the constant $k_3^{(1)}(y,t) = K_1$.  For a fixed $t$, the right-hand side of the implicit relation is a continuously differentiable and strictly decreasing function of $\hat{t}^{(3)}$ over the relevant interval, spanning the entire characteristic region. Thus, by the implicit function theorem, there exists a unique $\hat{t}^{(3)}=\hat{t}^{(3)}(y,t)$ that fixes the characteristic.

    This solution describes the continuous deformation of a local Y-soliton with phase parameters $(K_1, K_2, K_3)$ with a disturbance on the stem leg $j=2$ to a local Y-soliton with parameters $(K_1, K_2, K_3')$. An example numerical evolution of the KP equation with initial conditions corresponding to this case is shown in Fig.~\ref{fig:fastRWfromj2}.
    \begin{figure}
        \centering
        \includegraphics[width=0.98\linewidth]{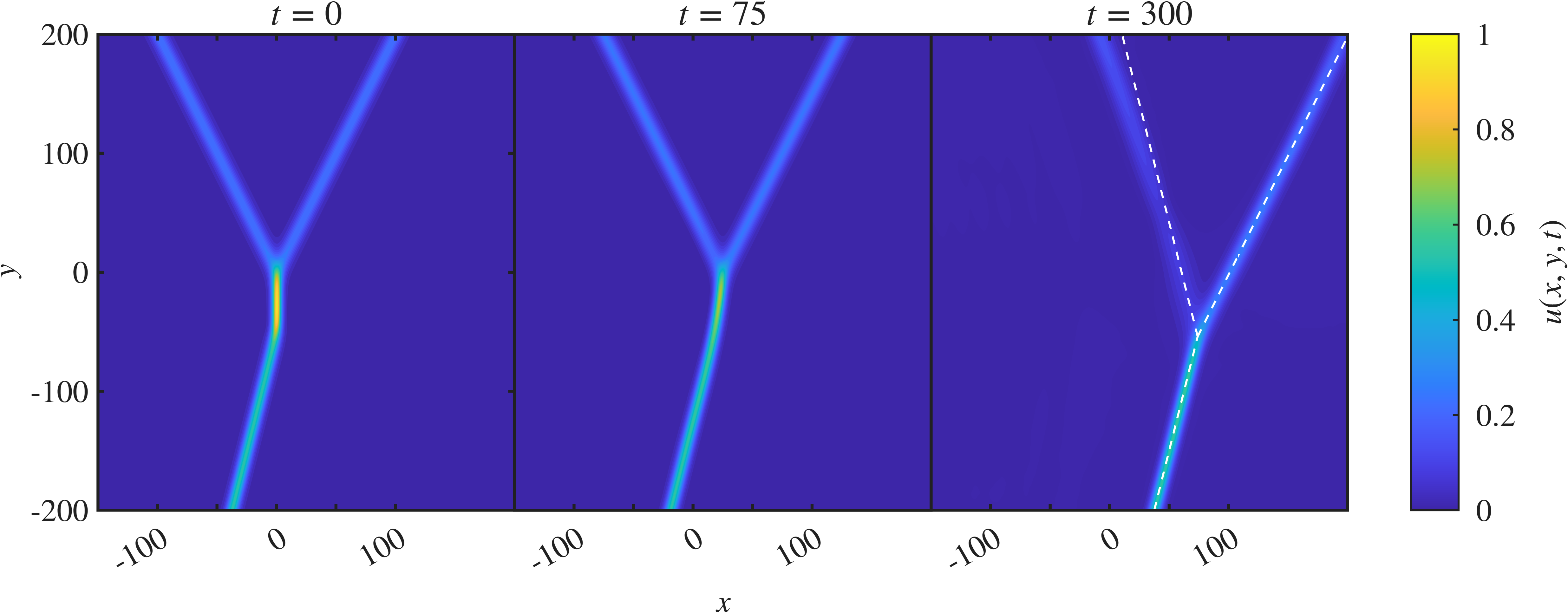}
        \caption{Numerical simulation of KP corresponding to a fast-RW from leg $j = 2$ to leg $j =1$. The relevant parameters are $K_1=-1, K_2=0, K_3'=0.5, K_3=1$, and $y_0=-50$. The dashed line shows an unmodulated Y-soliton with parameters of $(k_1, k_2, k_3)=(-1, 0, 0.5)$ at $t=300$, demonstrating the long time asymptotic state in the vicinity of the vertex.  See section \ref{sec:direct-numer-simul} for numerical details.} 
        \label{fig:fastRWfromj2}
    \end{figure}

    \subsubsection{Case $K_1 < K_3' \le K_2 < K_3$}
    \label{sec:leg2AsymptoticSingular}

    We now consider the special case of a RW from leg $j = 2$ to $j = 1$ with the parameter ordering
    \begin{equation}\label{eq:fastRWj2OrderingSpecial}
        K_1<K_2=K_3'<K_3,
    \end{equation}
    for the initial condition \eqref{eq:fastRWj2IC} associated with system \eqref{eq:mod_param_k3}--\eqref{eq:mod_param_k1}. Note that \eqref{eq:fastRWj2OrderingSpecial} does not immediately violate the local ordering on leg \(j=2\) since, on that leg, the parameters are well-ordered
    \[
    k^{(2)}_1=K_1 < K_3' \le k^{(2)}_3.
    \]
    
    Although \(t_1\) in eq.~\eqref{FastRWj2t1} remains finite and positive, equation \eqref{FastRWj2t2} shows that \(t_2\to\infty\) as \(K_3'\to K_2\) from the left. Thus, the rarefaction wave originating on leg \(j=2\) reaches the vertex, but never fully passes through it.
    
    This can be understood by comparing characteristic speeds. The trailing edge of the fast rarefaction wave on leg \(j=2\) travels with speed
    \begin{equation}\label{eq:trailingSpeed}
        \lambda^{(2)}_+(K_1,K_3')
        =\frac{4}{3}K_3'+\frac{2}{3}K_1.
    \end{equation}
    On the other hand, for \(t_1<t<t_2\), substitution of \eqref{eq:2RWj2Vertex} into \eqref{eq:2RWSolnj2Onj2} yields the lower bound on $k_3^{(2)}$
    \begin{equation}\label{eq:k3BoundAtVertex}
        k^{(2)}_3(p(t),t)
        =K_2+\frac{3}{2}\sqrt{\frac{2y_0(K_3-K_2)}{3t}}
        \ge K_2,
        \qquad t\in[t_1,t_2].
    \end{equation}
    Using \eqref{eq:ODEDefOfP} together with the continuity condition \eqref{eq:boundaryCont}, we  obtain the lower bound
    \begin{equation}\label{eq:vertexBoundSingular}
        \dot{p}(t)
        = \frac{2}{3}\bigl(K_1+K_2+k^{(2)}_3(p(t),t)\bigr)
        \ge \frac{2}{3}\bigl(K_1+2K_2\bigr).
    \end{equation}
    When \(K_3'=K_2\), comparison of \eqref{eq:trailingSpeed} and \eqref{eq:vertexBoundSingular} gives
    \begin{equation}\label{eq:trailingSpeedSlowerThanVertex}
        \lambda^{(2)}_+(K_1,K_3')
        = \frac{4}{3}K_3'+\frac{2}{3}K_1
        \le \frac{2}{3}\bigl(K_1+2K_2\bigr)
        \le \dot{p}(t).
    \end{equation}
    Since the RW trailing edge moves no faster than the vertex, and since initially \(y_0<p(0) = 0\), the RW can never fully pass through the vertex. In particular, the boundary data for $k_3^{(1)}(y,t)$ on leg \(j=1\) never drops below \(K_2\) but, using \(t_2\to\infty\), equation \eqref{eq:k3BoundAtVertex} shows that
    \[
    k^{(1)}_3(p(t),t)\to K_2
    \qquad\text{as } t\to\infty.
    \]
    This has some interesting implications.  On leg \(j=1\) the phase-parameter pair at the vertex asymptotically approaches
    \[
    (k^{(1)}_2,k^{(1)}_3)=(K_2,K_3)\to (K_2,K_2),
    \]
    which corresponds to vanishing amplitude. Thus leg \(j=1\) is \textit{asymptotically separates} from the other two legs $j = 2$ and $j =3$, which themselves form a modulated line soliton. In this sense, the long time limit is essentially an effective reduction from a three-parameter modulated Y-soliton to a two-parameter, $(K_1,K_2)$, modulated line soliton plus an asymptotically separated partial soliton representing the remnant of leg $j = 1$ \cite{ryskamp_evolution_2021}.  Note that since this reduction only happens as $t\to\infty$, the local ordering at the vertex is never violated which means that it falls under Theorem \ref{th:nonlinear_well_posed}, and accordingly, at any finite time, near the vertex, the structure is still locally a Y-soliton.  The main point is that the solution asymptotically approaches separated solutions as $t \to \infty$.  Also, as $t \to \infty$, the vertex velocity tends to one of the characteristic velocities $\dot{p}(t) \to \lambda_+^{(2)}(K_1,K_2)$ of the asymptotically separated, modulated line soliton. A numerical simulation illustrating this behavior is shown in Figure \ref{fig:singularFastRWj2}. 
    
    \begin{figure}[h!]
        \centering
        \includegraphics[width=0.98\linewidth]{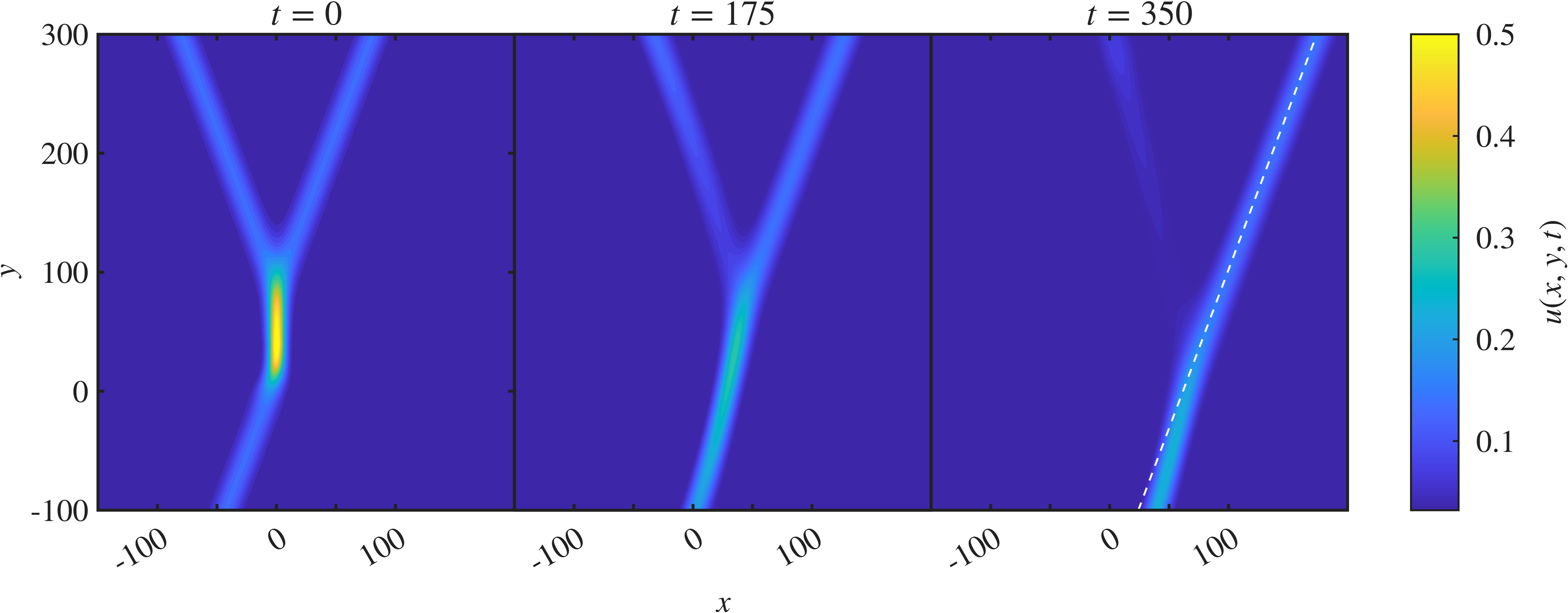}
        \caption{Long-time behavior in the special case \(K_3'=K_2\), with parameters $K_1=-0.75, K_2=K_3'=0, K_3=0.75$, and $y_0=-50$. Numerically, the amplitude of leg \(j=1\) near the vertex approaches zero and becomes comparable to the background noise level. The dashed line shows an unmodulated line soliton with parameters of $(k_1, k_2)=(-0.75, 0)$ at $t=350$ to which the solution locally asymptotes.}
        \label{fig:singularFastRWj2}
    \end{figure}
    
    Finally, note that the bounds \eqref{eq:k3BoundAtVertex} and \eqref{eq:trailingSpeedSlowerThanVertex} continue to hold when \(K_3'<K_2\), provided \(K_1<K_3'\) so that local ordering on leg \(j=2\) is preserved. The vertex effectively shields leg \(j=1\) from the potentially order-violating portion of the rarefaction generated on leg \(j=2\).

    \subsection{Leg $j = 3$}\label{sec:j3RW}

    For initial modulations of leg $j = 3$, both RWs will interact with the vertex and we analyze them both below.

    \paragraph{slow-RW}
    The calculation of this case is similar to that of the fast-RW initiated on leg $j = 2$ because the wave is transferred from one leg to another. We solve System \eqref{eq:mod_param_k3}--\eqref{eq:mod_param_k1} with the initial conditions
    \begin{equation}
        \begin{aligned}
                &k^{(j)}_3(y,0) = K_3, \quad j\in \{1,2 \}, \quad (-1)^j y < 0, \\
                &k^{(j)}_2(y,0) = K_2, \quad j\in \{1,3\}, \quad y > 0, \\
                &k^{(2)}_1(y,0) = K_1, \quad y < 0\\
                &k^{(3)}_1(y,0)=  \begin{cases}
                K_1& 0<y<y_0, \\
                K_1'& y>y_0 ,
                \end{cases}
        \end{aligned}
        \label{eq:slowRWj3IC}
    \end{equation}
    \begin{figure}[h!]
        \centering
        \includegraphics{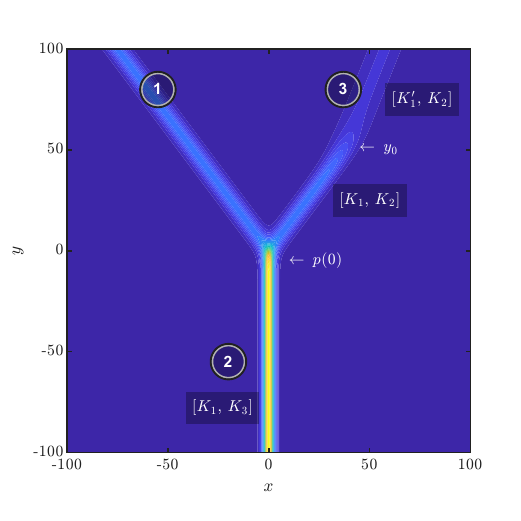}
        \caption{Sketch of the initial conditions \eqref{eq:slowRWj3IC} for the slow-RW initiated on leg $j = 3$.  The ensuing modulation is confined to legs $j = 3$ and $j = 2$.}
        \label{fig:slowRWj3IC}
    \end{figure}
    where we enforce $K_1<K_1'\le K_2<K_3$ to ensure an admissible slow-RW. See Fig.~\ref{fig:slowRWj3IC} for a schematic of the initial data \eqref{eq:slowRWj3IC}.
    Since $k_1^{(2,3)}$ are the variables evolving over time, the primary equations of interest are \eqref{eq:mod_param_k1_leg3} and \eqref{eq:mod_param_k1_leg2} because the remaining equations are trivially satisfied (the solution is constant) for all $y$ in the domain and $t > 0$.
    The solution for $y>p(t)$ on leg $j = 3$ is given as
    \begin{equation}
        k^{(3)}_1(y, t) = \begin{cases}
        K_1 & \text{if } y-y_0 < \lambda^{(3)}_-(K_1,K_2) t, \\
        \frac{3(y-y_0)}{4t} - \frac{1}{2}K_2 & \text{if } \lambda^{(3)}_-(K_1,K_2) t< y-y_0 < \lambda^{(3)}_-(K_1',K_2) t, \\
       K_1' & \text{if } y-y_0 > \lambda^{(3)}_-(K_1',K_2) t ,
        \end{cases}
    \end{equation}
    where
    \begin{equation}
        p(t) = \begin{cases}
        \frac{2}{3} \left(K_1 +K_2 + K_3 \right)t, & t < t_1, \\
         \left(1-2\sqrt{\frac{2(K_3 -K_1)}{3y_0}t}\right)y_0 + \left(\frac{2}{3}K_2  + \frac{4}{3}K_3\right) t,  & t_1 <t< t_2, \\
         \frac{2}{3}  (K_1' +K_2 +K_3 )t + y_2,& t > t_2,
        \end{cases}
    \end{equation}
    with $t_1= \frac{3y_0}{2(K_3 -K_1)}$, $t_2= t_1 \left ( \frac{K_3-K_1}{K_3  -K_1'}\right)^2$, and $y_2=-y_0\frac{K_1' - K_1}{K_3 - K_1'} < 0$.
    This determines the solution on leg $j=2$ by tracing back a slow characteristic to the boundary
    \eqref{eq:boundaryCont}
    \begin{subequations}
        \begin{equation}
            k^{(2)}_1(y,t)= \begin{cases}
            K_1 & y \leq (\frac{4}{3}K_1 + \frac{2}{3}K_3 )t, \\
            f(\hat{t}^{(2)}) & p(t)>y > (\frac{4}{3}K_1 + \frac{2}{3}K_3 )t, 
            \end{cases} 
        \end{equation}
        \text{where $f(\hat{t}^{(2)})$ results from the boundary data}
        \[
         f(\hat{t}^{(2)})  = \begin{cases}
            \frac{3(p(\hat{t}^{(2)})-y_0)}{4\hat{t}^{(2)}} - \frac{1}{2}K_2 & \text{if } t_1 <\hat{t}^{(2)}< t_2 \\
            K_1' & \text{if } \hat{t}^{(2)} > t_2
            \end{cases} 
        \]
        \end{subequations}
        and the variable $\hat{t}^{(2)}=\hat{t}^{(2)}(y,t)$ is the foot of the characteristic on leg $j=2$ and is uniquely defined by the implicit function theorem as
        \[
        y  = (\frac{4}{3}f(\hat{t}^{(2)}) + \frac{2}{3}K_3)(t-\hat{t}^{(2)}) +p(\hat{t}^{(2)}) .
        \]
    \begin{figure}
        \centering
        \includegraphics[width=0.98\linewidth]{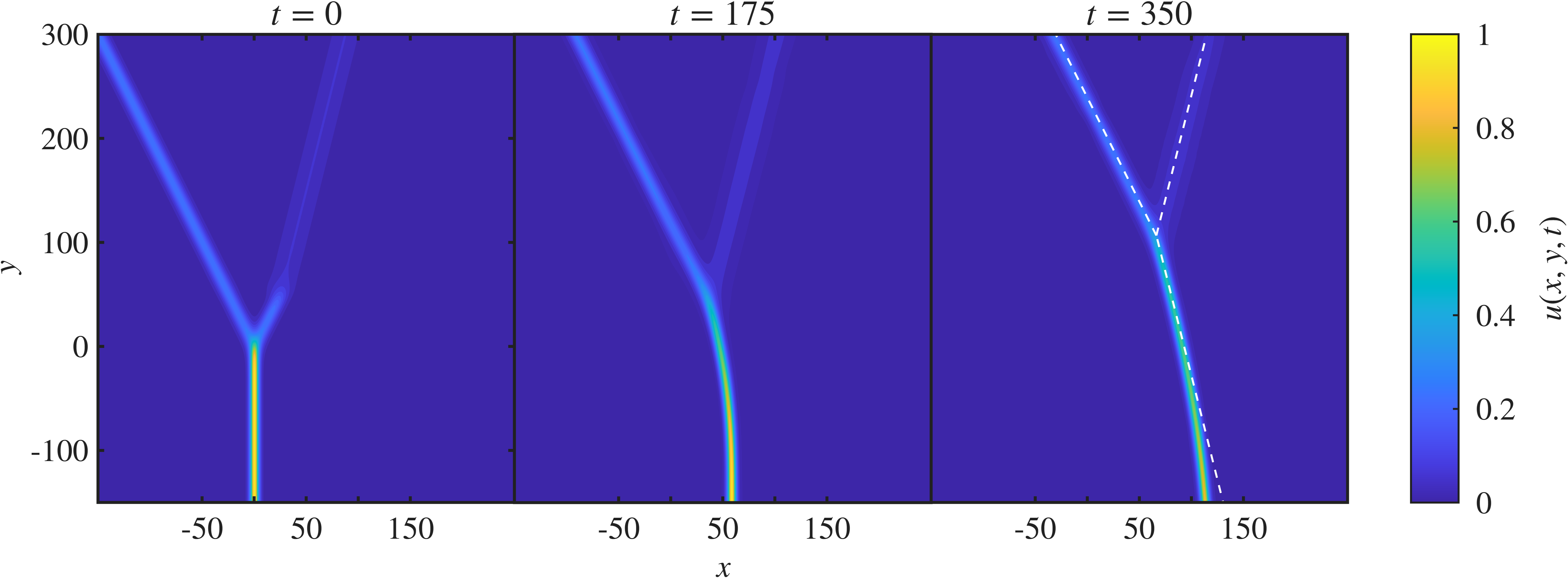}
        \caption{Numerical simulation of a slow-RW initiated on leg $j = 3$ that fully transmits to leg $j = 2$. The relevant parameters are $K_1=-1,K_1' =-0.5, K_2=0, K_3=1$, and $y_0=50$. The dashed line shows an unmodulated Y-soliton with parameters of $(k_1, k_2, k_3) = (-0.5, 0, 1)$ at $t=350$ that the solution approaches near the vertex.}
        \label{fig:slowRWfromj3}
    \end{figure}
    The solution on leg $j = 2$ for $0<t<t_1$ and $y < p(t)$ is the constant $k_1^{(2)}(y,t) = K_1$.

     This solution describes a RW that propagates from leg $j = 3$ to leg $j = 2$, which converts the Y-soliton with phase parameters of $(k_1,k_2,k_3)=(K_1, K_2, K_3)$ to a Y-soliton with parameters $(k_1,k_2,k_3) = (K_1', K_2, K_3)$.  See Fig.~\ref{fig:slowRWfromj3} for a representative numerical simulation of the KP equation.

    A notable special case of this system is when we choose $K_1'=K_2$.  The RW is still admissible and passes through the vertex. 
    This corresponds to a partial soliton for leg $j=3$.  In particular, for $y>y_0$, the amplitude of the leg is zero. The previous analysis for the mathematical solution is still valid for this case. However, note that in long time, instead of replacing one parameter with another parameter, that is transforming $(K_1,K_2,K_3) \to (K_1', K_2, K_3)$, we have instead `deleted' a parameter, going from a 3-parameter system to a 2-parameter system $(K_1,K_2,K_3)\to(K_2, K_3)$. Again, this suggests a change in the geometry of the solution. Instead of the long time solution being in the form of a Y-soliton with altered parameters, the solution tends toward a line soliton since legs $j = 1$, $j = 2$ are both parameterized by $(K_2, K_3)$ in long time whereas leg $j = 3$ is effectively deleted.  The leg deletion occurs in finite time when the upper extremity of the slow-RW reaches the vertex.
    Note that this exact leg deletion case violates strict ordering $k_1 < k_2 < k_3$ at the vertex and thus correspondingly does not satisfy Theorem \ref{th:nonlinear_well_posed}. In fact, the vertex velocity $\dot{p}(t) \to \lambda_-^{(1)}$ degenerates into the slow characteristic velocity $\lambda_-^{(1)}$ at the lower edge of a modulated line soliton with parameters $(K_2, K_3)$.
    \begin{figure}
        \centering
        \includegraphics[width=0.98\linewidth]{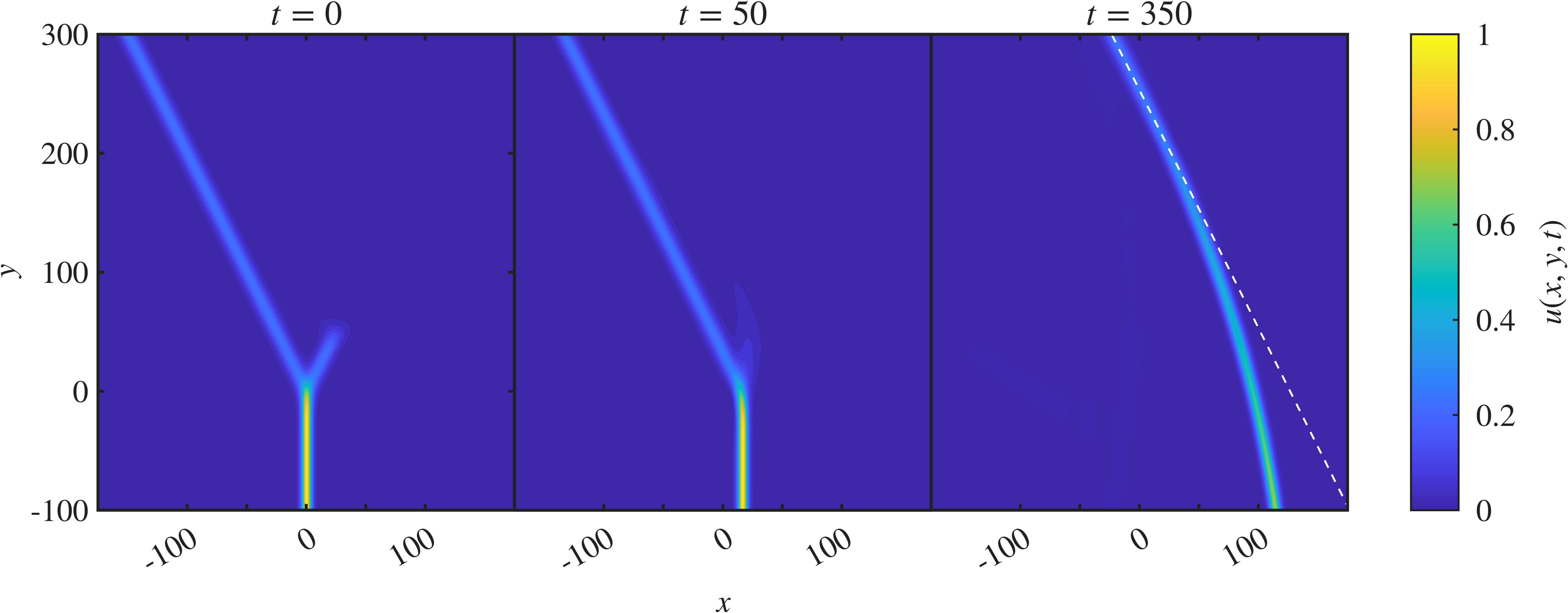}
        \caption{This figure displays the time evolution from a numerical simulation of a singular slow rarefaction starting from leg $j=3$ and transmitting to leg $j=2$. The relevant parameters are $K_1=-1,K_1' =K_2=0, K_3=1$, and $y_0=50$. The dashed line shows an unmodulated line soliton with parameters of $(k_1, k_2) = (0, 1)$ at $t=350$ for comparison.
        }
        \label{fig:slowRWfromj3Del}
    \end{figure}

    Figure \ref{fig:slowRWfromj3Del} displays a numerical simulation of the KP equation for the slow-RW from leg $j = 3$ to $j=2$ for the special parameter value $K_1' = K_2$ that eventually leads to the deletion of leg $j = 3$.
    
    \paragraph{fast-RW}

    The case of the fast-RW from leg $j=3$ to $j=1$ is the most complex case we consider because gradient catastrophe in the modulation equations on leg $j=1$ after the RW has encountered the vertex occurs.  A theorem that applies to this case is proven in Appendix \ref{sec:appendix_grad_catastrophe}.  A necessary condition in the corollary to the theorem demonstrates that all other RW cases for Y-soliton modulation do not lead to gradient catastrophe.  Care must be taken to suitably regularize the singularity, which results in the addition of a new, growing soliton leg analogous to \cite{ryskamp_modulation_2022,han_regularizations_2026}.
        
    For the fast-RW from leg $j=3$ to leg $j=1$, consider \eqref{eq:mod_param_k3}--\eqref{eq:mod_param_k1} with initial conditions
    \begin{equation}
        \begin{aligned}
                &k^{(j)}_3(y,0) = K_3, \quad j\in \{1,2 \}, \quad (-1)^{j}y < 0, \\
                &k^{(1)}_2(y,0) = K_2, \\
                &k^{(3)}_2(y,0)=  \begin{cases}
                K_2& 0<y<y_0, \\
                K_2'& y>y_0
                \end{cases},\\
                &k^{(j)}_1(y,0) = K_1, \quad j\in \{2,3 \}, \quad (-1)^j y < 0,
        \end{aligned}
        \label{eq:fastRWj3IC}
    \end{equation}
    where the ordering 
    \begin{equation}\label{eq:fastRWj3Ordering}
        K_1<K_2<K_2'<K_3,
    \end{equation}
    ensures an admissible fast-RW initiated on leg $j=3$. See Fig.~\ref{fig:fastRWj3IC} for a schematic of the initial data \eqref{eq:fastRWj3IC}.  Since $k_2$ is the variable evolving over time, the primary equation of interest is \eqref{eq:mod_param_k2}.
    
    \begin{figure}
        \centering
        \includegraphics{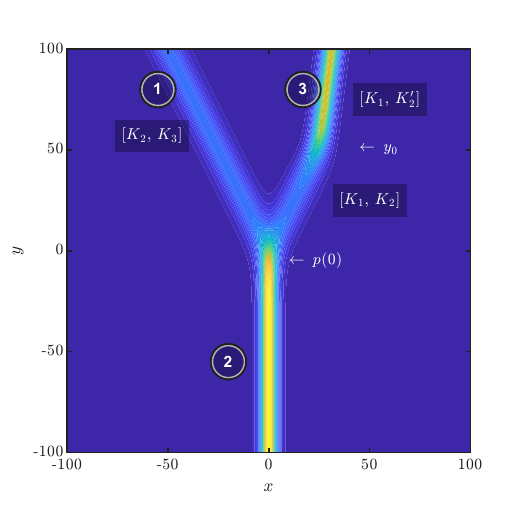}
        \caption{Sketch of initial conditions \eqref{eq:fastRWj3IC} for the fast-RW initiated on leg $j=3$. The ensuing modulation is confined to legs $j=3$ and $j=1$. }
        \label{fig:fastRWj3IC}
    \end{figure}

    We begin as we did before. The solution for leg $j=3$ is
    \begin{equation}
        k^{(3)}_2(y, t) = \begin{cases}
        K_2 & \text{if } y-y_0 < \lambda^{(3)}_+(K_1,K_2) t, \\
        \frac{3(y-y_0)}{4t} - \frac{1}{2}K_1 & \text{if } \lambda^{(3)}_+(K_1,K_2) t < y-y_0 < \lambda^{(3)}_+(K_1,K_2') t , \\
       K_2' & \text{if } y-y_0 > \lambda^{(3)}_+(K_1,K_2') t,
        \end{cases}
    \end{equation}
    for $y>p(t)$ where
    \begin{equation}
        p(t) = \begin{cases}
        \frac{2}{3} \left(K_1 + K_2 + K_3 \right)t, & 0 < t < t_1, \\
         \left(1-2\sqrt{\frac{2(K_3 -K_2)}{3y_0}t}\right)y_0 + \left(\frac{2}{3}K_1  + \frac{4}{3}K_3\right) t,  & t_1 <t< t_2, \\
         \frac{2}{3}  (K_1+K_2' +K_3 )t + y_2,& t > t_2.
        \end{cases}
    \end{equation}
    
    with 
    \begin{equation}\label{eq:t1t2New}
        t_1=\frac{3y_0}{2(K_3 -K_2)},  \quad t_2=t_1 \left ( \frac{K_3-K_2}{K_3  -K_2'}\right)^2,
    \end{equation} and $y_2=-y_0\frac{K_2' - K_2}{K_3-K_2'} < 0$.  

    Using backward characteristics, we can express the solution on leg $j=1$ in terms of the data at the boundary
    \begin{equation}\label{eq:boundaryDataDefinition}
        k_2^{(1)}(y,t) = f(\hat{t}^{(1)})  = \begin{cases}
        K_2 & \text{if } \hat{t}^{(1)} < t_1, \\
        \frac{3(p(\hat{t}^{(1)})-y_0)}{4\hat{t}^{(1)}} - \frac{1}{2}K_1 & \text{if } t_1 <\hat{t}^{(1)}< t_2, \\
        K_2' & \text{if } \hat{t}^{(1)} > t_2,
        \end{cases} 
    \end{equation}
    where $\hat{t}^{(1)}=\hat{t}^{(1)}(y,t)$ is the propagation along the slow characteristic velocity of leg $j = 1$
    \begin{equation}\label{eq:chars}
        y  = (\frac{4}{3}f(\hat{t}^{(1)}) + \frac{2}{3}K_3)(t-\hat{t}^{(1)}) +p(\hat{t}^{(1)}).
    \end{equation}
    Figure \ref{fig:charSketch}(a) displays characteristics in the $y$-$t$ plane for this solution as gray (where $k_2^{(1)}(y,t) = K_2$, $y > p(t)$, $t < t_1$) and orange (the narrow region where $K_2' > k_2^{(1)} > K_2$, $y > p(t)$, $t_1 < t < t_{\rm break}$) curves.
    
    While the RW is passing through the boundary ($t_1<t<t_2$), $f'(t) > 0$, a necessary condition for wavebreaking in Corollary \ref{cor:noShockCriterion}. Now we proceed to show sufficient conditions for gradient catastrophe.

    The characteristic velocity in equation \eqref{eq:mod_param_k2_leg1} is $c(k_2^{(1)}) = \lambda^{(1)}_-(k^{(1)}_2, K_3)=\frac{4}{3}k^{(1)}_2 + \frac{2}{3}K_3$. In order to apply Theorem \ref{th:shockCondition}, the boundary velocity $\dot{p}(t)$ given in \eqref{eq:ODEDefOfP}, combined with the bounds $K_1<K_2\leq k_2(y,t)\leq K_2' < K_3$ yield the following bounds
    \begin{equation}
        \frac{2}{3}(K_1+K_2+K_3)\le\dot{p}(t)\le \frac23(K_1+K_2'+K_3)=V_+.
    \end{equation}
    Since \(K_1<K_2'\) implies
    $
    \frac{K_1+K_2'}{2}<K_2'
    $,
    and since \(f\) increases continuously up to the value \(K_2'\), there exists \(\hat{t}^{(1)}\in(t_1,t_2)\) such that
    \[
    \frac{K_1+K_2'}{2}<f(\hat{t}^{(1)}).
    \]
    This implies that 
    \[
     \frac23(K_1+K_2'+K_3)<\frac43 f(\hat{t}^{(1)})+\frac23 K_3 = c(f(\hat{t}^{(1)})),
    \]
    and that we have found the existence of a $\hat{t}^{(1)}$ such that $c(f(\hat{t}^{(1)})) > V_+$. Furthermore, $f(t)$ is increasing, so, in particular, $f'(\hat{t}^{(1)}) > 0$, the sufficient conditions from \eqref{eq:aboveVertexShockCondition} in Theorem~\ref{th:shockCondition} are satisfied.  Thus we conclude that gradient catastrophe occurs in the region $y>p(t)$. 

    Note that function $\hat{t}^{(1)} = \hat{t}^{(1)}(y,t)$ becomes multivalued after wave breaking. For notational purposes, the branch of $\hat{t}^{(1)}$ we refer to when calculating the breaking time is the branch that originates from the moving boundary $p(t)$. See \ref{fig:charSketch} for a visual of this and Appendix \ref{sec:appendix_grad_catastrophe} for more details.

        \begin{figure}
        \includegraphics{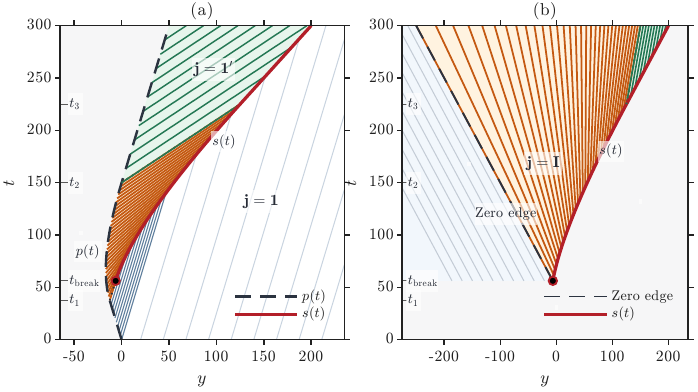}
        \caption{The characteristic plane for a fast-RW case initiated from leg $j=3$ to $j=1$. Panel (a) displays leg $j=1$ and $j=1'$ separated by $s(t)$ in the characteristic plane. Panel (b) displays leg $j=I$ forming from $s(t)$. The gray characteristics in (a) propagate the data $k_2^{(1)} = K_2$; the orange characteristics in (a,b) propagate changes in $K_2 < k_2^{(1',I)} < K_2'$, and the green characteristics in (a,b) represent the stabilized data $k_2^{(1',I)} = K_2'$ corresponding to the states transported after the simple wave has completely passed through $p(t)$.  In (a), the thick red curve depicts the new vertex position $y = s(t)$ while the dashed curve is the original vertex $y = p(t)$.  In (b), the partial soliton terminates at zero amplitude, denoted by the dashed curve, with the red curve displaying the new vertex position $y = s(t)$.  In (a) and (b), $y = s(t)$ connect the characteristics from leg $j = 1'$ to leg $j = I$.}
        \label{fig:charSketch}
    \end{figure}

    Now we determine the breaking time $T$ in equation \eqref{eq:breaking_time} from Theorem \ref{th:shockCondition}.  Insertion of $c(k)$, $f(t)$, and $p(t)$ for the fast-RW, we find that it is an increasing function of $\hat{t}^{(1)}$. Consequently, the breaking time $T$, minimized over $\hat{t}^{(1)}$, is achieved when $\hat{t}^{(1)} = t_1$, i.e., the initiation of RW-boundary interaction.
    Inserting this into equation \eqref{eq:shockTime} yields the breaking time
    \begin{equation}\label{eq:tbreak}
        t_{\rm break}
        =\frac{3y_0(K_3-K_1)}
        {2(K_3-K_2)^2}.
    \end{equation}
    and location
    \begin{equation}
        \label{eq:ybreak}
        y_{\rm break} = \left ( \frac{4}{3} f(t_1) + \frac{2}{3}K_3 \right ) (t_{\rm break} - t_1) + p(t_1) ,
    \end{equation}
    that is identified in Fig.~\ref{fig:charSketch}(a) as the black dot.
    We observe that the breaking time is independent of $K_2'$, only depending on the existence of  positive increase of $k_2^{(3)}$ at $y = y_0$.  We will continue to utilize backward characteristics in our calculations.  Anticipating the regularization of gradient catastrophe by the introduction of a new leg $j = 1'$, we refer to the time $\hat{t}^{(1')}(y,t)$ from a point $(y,t)$, $y > p(t)$ and $t > t_{\rm break}$ as the backward characteristic to the vertex
    \begin{equation}\label{eq:newt1}
        y = \lambda_-(k_2^{(1')},K_3)(t-\hat{t}^{(1')}) + p(\hat{t}^{(1')}),
    \end{equation}
    that remains on leg $j = 1'$.

    \begin{figure}
        \centering
        \includegraphics{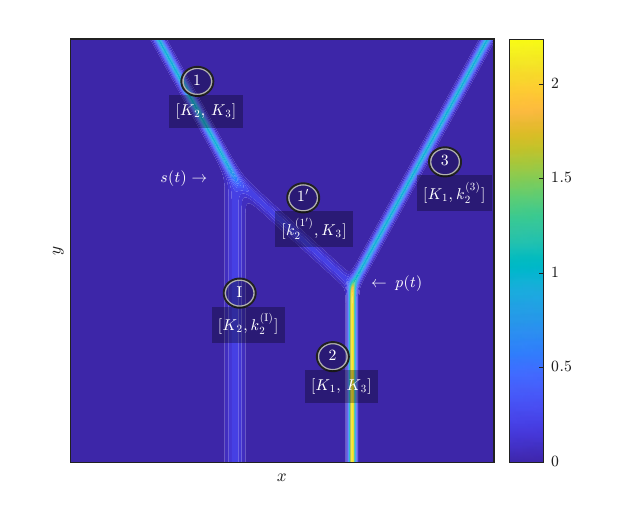}
        \caption{Leg labeling sketch for the post breaking modulation system. Leg $j=1$ from the previous labeling in Fig.~\ref{fig:YSoliLabeling} has been partitioned into three segments $j=1$, $j=1'$, and $j=I''$ that (locally) compose the inverted Y-soliton.}
        \label{fig:shockCreationDiagram}
    \end{figure}

    Now we regularize breaking by the introduction of a new leg that locally corresponds to an inverted Y-soliton (see Fig.~\ref{fig:twoYs}(b)) subject to analogous continuity conditions \eqref{eq:boundaryCont} and vertex dynamics \eqref{eq:ODEDefOfP} of the Y-soliton modulation system. To achieve this, we introduce the augmented labeling scheme depicted in Fig.~\ref{fig:shockCreationDiagram}.  While we could add the modulation equations for the newly formed inverted Y-soliton to the modulation equations for the original, upright Y-soliton, only the modulation parameters $k_2^{(1',I,3)}$ vary across the entire structure from leg $j=3$ through leg $j=1'$ to leg $j=I$.  All other modulation variables are constant and denoted in Fig.~\ref{fig:shockCreationDiagram}.

    Associated with the inverted Y-soliton is its vertex located at $y = s(t)$ for $t > t_{\rm break}$.  From this vertex are three distinct soliton legs $j\in \{1,1',I$\} as shown in Fig.~\ref{fig:shockCreationDiagram}. The equations for $k_2^{(3,1',I)}$ describing this modulation are, for $t > t_{\rm break}$, 
    \begin{subequations}\label{eq:new_mod_param_k2}
    \begin{align}
        \label{eq:new_mod_param_k2_leg3}
        \partial_t k_2^{(3)}
        +\lambda_+^{(3)}(k_1^{(3)},k_2^{(3)})
        \partial_y k_2^{(3)}
        &=0, \quad y>p(t),\\
        \label{eq:new_mod_param_k2_leg1}
        \partial_t k_2^{(1')}
        +\lambda_-^{(1')}(k_2^{(1')},k_3^{(1')})
        \partial_y k_2^{(1')}
        &=0, \quad p(t)<y<s(t),\\
        \label{eq:new_mod_param_k2_new}
        \partial_t k_2^{(I)}
        +\lambda_+^{(I)}(K_2,k_2^{(I)})
        \partial_y k_2^{(I)}
        &=0, \quad y<s(t),
    \end{align}
    \end{subequations}
    where
    \begin{equation}\label{eq:augmentedSpeedDefs}
        \begin{aligned}
            \lambda_+^{(3)}(k_1^{(3)},k_2^{(3)})
            &=\frac23 k_1^{(3)}+\frac43 k_2^{(3)},\\
            \lambda_-^{(1')}(k_2^{(1')},k_3^{(1')})
            &=\frac43 k_2^{(1')}+\frac23 k_3^{(1')},\\
            \lambda_+^{(I)}(K_2,k_2^{(I)})
            &=\frac23 K_2+\frac43 k_2^{(I)},
        \end{aligned}
    \end{equation}
    subject to the analogous continuity conditions \eqref{eq:boundaryCont}, which include 
    \begin{equation}
         \label{eq:inverted_Y_continuity}
        k_2^{(1')}(s(t),t) = k_2^{(I)}(s(t),t), \quad k_2^{(1')}(p(t),t) = k_2^{(3)}(p(t),t),
    \end{equation}
    and analogous vertex dynamics \eqref{eq:ODEDefOfP}
    \begin{equation}
\label{eq:inverted_Y_vertex_speed}
        \dot{s} = \frac{2}{3}(K_2 + k_2^{(1')}(s,t)+K_3), \quad s(t_{\rm break}) = y_{\rm break} .
    \end{equation}
    The initial condition on leg $j = I$ is
     \begin{equation}
     \label{eq:Iprime_k2_IC}
    k^{(I)}_2(y,t_{\rm break})= K_2, \quad y < s(t_{\rm break})
    \end{equation}
    which corresponds to zero amplitude on that leg.  For proper bookkeeping, we acknowledge the label changing with the data
    \begin{equation}
        \label{eq:1prime_k2_IC}
        k_2^{(1')}(y,t_{\rm break}^+) = k_2^{(1)}(y,t_{\rm break}^-), \quad p(t_{\rm break}) < y < s(t_{\rm break}) .
    \end{equation}

We can solve the system \eqref{eq:new_mod_param_k2} subject to \eqref{eq:inverted_Y_continuity}, \eqref{eq:inverted_Y_vertex_speed}, \eqref{eq:Iprime_k2_IC}, \eqref{eq:1prime_k2_IC} in a similar vein as before, using the method of characteristics.  The following boundary data for equation \eqref{eq:new_mod_param_k2_new}
    \begin{equation}
        k^{(I)}_2 \left(s(t),t\right) = g(t) = f(\hat{t}^{(1')}(s(t),t)), \quad t > t_{\rm break},
    \end{equation}
    where $f$ is defined in \eqref{eq:boundaryDataDefinition}.
        Solving this with the method of characteristics yields
        \begin{equation}
            k^{(I)}_2 (y,t)= \begin{cases}
            K_2 & y \leq s(t_\text{break})+2K_2(t-t_{\rm break})\\
            g(\hat{t}^{(I)}(y,t)) & s(t_\text{break})+2K_2(t-t_{\rm break}) < y < s(t)
            \end{cases}
        \end{equation}
        where $\hat{t}^{(I)}(y, t)$ is defined implicitly by
        \begin{equation}\label{eq:MOCdefinitionForNewLeg}
            y  = \left ( \frac{4}{3}g(\hat{t}^{(I)}) + \frac{2}{3}K_2 \right )(t-\hat{t}^{(I)}) +s(\hat{t}^{(I)}) .
        \end{equation}

        The solution can be simplified. Define $\tau(t) =\hat{t}^{(1')}(s(t),t)$ according to $k_2(p(\tau),\tau) = g(t)$, which is the time that the value $g(t)$ first reaches the original vertex $p(t)$. That is the characteristic corresponding to $g(t)$ reached the original Y-soliton vertex at $(y,t)=(p(\tau(t)),\tau(t))$.  We can obtain an explicit expression for $\tau(t)$ by using both the vertex dynamics \eqref{eq:inverted_Y_vertex_speed} and the implicit definition of $\hat{t}^{(1')}$ in \eqref{eq:newt1}. Evaluating $\hat{t}^{(1')}(s(t), t)=\tau(t)$ in expression \eqref{eq:newt1} and differentiating provides another representation of $\dot{s}(t)$. Equating this with the new vertex dynamics \eqref{eq:inverted_Y_vertex_speed} results in an ODE for $\tau(t)$, which upon solving with suitable conditions yields 
        \begin{equation}\label{eq:tauTot}
            \tau(t) =\frac{K_3-K_2}{K_3-K_1}t.
        \end{equation}
        Now, consider
        \begin{equation}\label{eq:newVertexAsDistanceTraveled}
            s(t)=y_0+
            \lambda_+^{(3)}(K_1,g(t))\tau(t)+
            \lambda_-^{(1')}(g(t),K_3)(t-\tau(t)),
        \end{equation}
        representing the position of the new vertex $s(t)$ as the sum of the distances traveled on leg $j=3$ and $j=1'$. Utilizing the corresponding definition of respective velocities and eq.~\eqref{eq:tauTot} for $\tau(t)$, eq.~\eqref{eq:newVertexAsDistanceTraveled} reduces to 
        \begin{equation}\label{eq:miracleResult}
            s(t)=y_0+
            \left(\frac43g(t)+\frac23K_2\right)t=y_0+t\lambda^{(I)}_+\left( K_2, g(t) \right).
        \end{equation} 
        In obtaining eq.~\eqref{eq:miracleResult}, we find the ``averaging identity''
        \begin{equation}
            \label{eq:averagingIdentity}
            \frac{K_1\tau(t)+K_3(t-\tau(t))}{t}=K_2.
        \end{equation}
        Evaluating the slow rarefaction solution initiated on leg $j=3$ at the vertex $y = p(\tau)$, we obtain 
        \begin{equation}
            \label{eq:goft}
            g(t) = k_2^{(3)}(p(\tau(t)),\tau(t)), \quad t > t_{\rm break},
        \end{equation}
        which uniquely specifies all terms in eq.~\eqref{eq:miracleResult}.
        The characteristic on leg $j = 3$ that propagates with the value $g(t)$ has velocity $\lambda_+(K_1,g(t))$ and its continuation on leg $j = 1'$ has velocity $\lambda_-(g(t),K_3)$.  Despite the fact that these two velocities differ, the averaging identity \eqref{eq:averagingIdentity} allows us to express the total \emph{displacement} from the initial point of the solution value $g(t)$ along these two characteristics at time $t$ as the average characteristic velocity multiplied by time $\lambda_+(K_2,g(t)) t$.  This explains eq.~\eqref{eq:miracleResult} and provides intuition that allows for algebraic simplification of the rest of the solution.

        Using \eqref{eq:miracleResult}, the characteristic equation \eqref{eq:MOCdefinitionForNewLeg} simplifies to
        \begin{align}
            y &= \lambda_+^{(I)}(K_2,g(\hat{t}^{(I)}))
                 (t-\hat{t}^{(I)})+s(\hat{t}^{(I)}) \notag\\
              &= y_0+\lambda_+^{(I)}(K_2,g(\hat{t}^{(I)}))t.
        \end{align}
        Since $k_2^{(I)}(y,t)=g(\hat{t}^{(I)})$, we can eliminate the implicit characteristic emission time $t^{(I)}$ and solve for $k_2^{(I)}(y,t)$ to obtain
        \[
         k_2^{(I)}(y,t)=-\dfrac12K_2+\dfrac{3(y-y_0)}{4t}
                \quad \text{for } y_0+2K_2t<y<s(t),
        \]
        Thus, our final solution for the new leg is simply expressed as the fast-RW
        \begin{equation}\label{eq:newLegExplicitSolution}
        \begin{aligned}
        k_2^{(I)}(y,t)
        &=
        \begin{cases}
        K_2,
        & y-y_0\leq \lambda_+^{(I)}(K_2,K_2)t,\\[4pt]
        \frac{3(y-y_0)}{4t}-\frac12K_2,
        & \begin{aligned}[t]
          \lambda_+^{(I)}(K_2,K_2)
          &< y-y_0< \min\bigl\{s(t)-y_0,\\
          &\quad\qquad\qquad \lambda_+^{(I)}(K_2,K_2')t\bigr\},
          \end{aligned}\\
        K_2',
        & \lambda_+^{(I)}(K_2,K_2')t
        \leq y-y_0<s(t)-y_0,\quad t>t_3.
        \end{cases}
        \end{aligned}
        \end{equation}
        for $t\geq t_{\rm break}$ and $y<s(t)$, where $t_3$ is the time at which the RW has completely passed through the new vertex.
        The remaining unknowns are provided by using eq.~\eqref{eq:tauTot} to obtain
        \begin{equation}\label{eq:newVertexPos}
        s(t)=
        \begin{cases}
        y_0+\left(\dfrac43K_3+\dfrac23K_2\right)t
        -2\sqrt{\dfrac{2y_0(K_3-K_1)t}{3}},
        & t_{\rm break}\leq t\leq t_3,\\[8pt]
        s(t_3)+\dfrac23(K_2+K_2'+K_3)(t-t_3),
        & t>t_3.
        \end{cases}
        \end{equation}
        The RW passes through the vertex $y = s(t)$ when the final RW state $K_2'$ reaches
        it, i.e., when $\tau(t_3)=t_2$, yielding
        \begin{equation}
        \label{eq:t3}
        t_3=\frac{K_3-K_1}{K_3-K_2}t_2
        =\frac{3y_0(K_3-K_1)}{2(K_3-K_2')^2}.
        \end{equation}

        Again, \eqref{eq:miracleResult} implies that all characteristics emitted from the new vertex as the RW passes through it, when extended backward, intersect at $(y_0,0)$. Thus, as seen by solution \eqref{eq:newLegExplicitSolution}, the new leg forms a centered rarefaction wave despite being generated by a shock at $t>t_{\rm break}$. Particularly, the result is as if leg $j=I$ was extended in $y$, and was given fast-RW data centered at $y=y_0$.  thus, leg $j = I$ can be viewed as a partial line soliton as defined in \cite{ryskamp_evolution_2021}. The corresponding amplitude and slope of the newly generated
        leg $j=I$ follow from the line-soliton relations
        \[
            a^{(I)}=\frac14\left(k_2^{(I)}-K_2\right)^2,
            \qquad
            q^{(I)}=\frac12\left(K_2+k_2^{(I)}\right).
        \]
        Substituting \eqref{eq:newLegExplicitSolution}, we obtain
        \begin{equation}\label{eq:newLegAmplitude}
        a^{(I)}(y,t)=
        \begin{cases}
        0,
        & y-y_0\leq 2K_2t,\\[4pt]
        \dfrac{9}{64t^2}(y-y_0-2K_2t)^2,
        & \begin{aligned}[t]
        2K_2t &< y-y_0\\
        &<\min\bigl\{s(t)-y_0,
        \lambda_+^{(I)}(K_2,K_2')t\bigr\},
        \end{aligned}\\[8pt]
        \dfrac14(K_2'-K_2)^2,
        & \begin{aligned}[t]
        \lambda_+^{(I)}(K_2,K_2')t
        &\leq y-y_0<s(t)-y_0,\\
        &t>t_3,
        \end{aligned}
        \end{cases}
        \end{equation}
        and
        \begin{equation}\label{eq:newLegSlope}
        q^{(I)}(y,t)=
        \begin{cases}
        K_2,
        & y-y_0\leq 2K_2t,\\[4pt]
        K_2+\dfrac{3}{8t}(y-y_0-2K_2t),
        & \begin{aligned}[t]
        2K_2t &< y-y_0\\
        &<\min\bigl\{s(t)-y_0,
        \lambda_+^{(I)}(K_2,K_2')t\bigr\},
        \end{aligned}\\[8pt]
        \dfrac12(K_2+K_2'),
        & \begin{aligned}[t]
        \lambda_+^{(I)}(K_2,K_2')t
        &\leq y-y_0<s(t)-y_0,\\
        &t>t_3.
        \end{aligned}
        \end{cases}
        \end{equation}
        Indeed, the new partial soliton leg exhibits a
        quadratically varying amplitude and a linearly varying
        slope at a given time. 
        
        The new leg $j=I$ is formally degenerate with zero
        amplitude for $t < t_{\rm break}$.
        For $t > t_{\rm break}$, the amplitude vanishes at its lower boundary,
        $y=y_0+2K_2t$, and increases toward the new vertex.
        For $t > t_3$, a constant-amplitude and constant-slope region
        develops adjacent to the vertex, corresponding to the
        line-soliton parameters $(K_2,K_2')$.
        
        For $t > t_3$, the vertex \eqref{eq:inverted_Y_vertex_speed} moves with constant speed
        \begin{equation}
            \label{eq:inverted_Y_vertex_position}
            s(t) = \frac{2}{3}(K_2+ K_2' + K_3)(t-t_3) + y_3
        \end{equation}
        where 
        $
        y_3=s(t_3)=y_0+\left(\frac43K_2'+\frac23K_2\right)t_3.
        $
        For $t>t_3$, the vertex propagates as an unmodulated Y-soliton vertex with soliton structure locally resembling a Y-soliton around the new resonant vertex $s(t)$.

    \begin{figure}
        \includegraphics[width=0.98\linewidth]{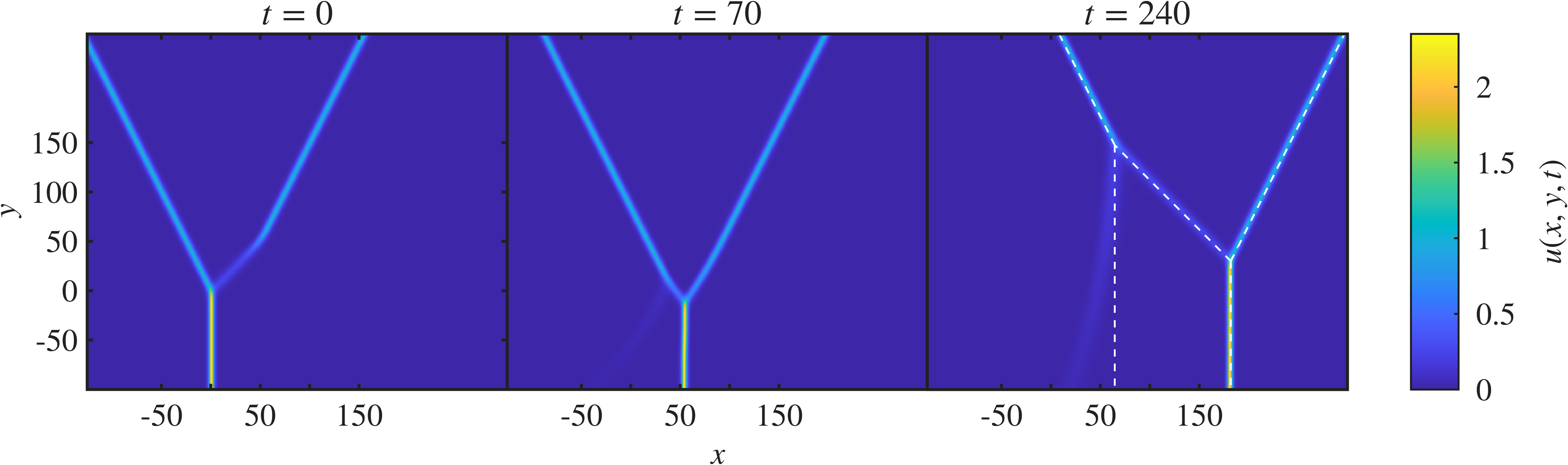}
        \caption{This figure displays the time evolution from a numerical simulation of a fast rarefaction starting from leg $j=3$ and transmitting to leg $j=1$. The relevant parameters are $K_1=-1,K_2=0, K_2'=0.75,K_3=1$, and $y_0=50$. The dashed line shows an unmodulated 3421 soliton line graph with parameters $(k_1,k_2,k_3, k_4) = (-1.5, -0.5, 0.5, 1.5)$ for comparison.}
        \label{fig:fastRWfromj3Smaller}
    \end{figure}

   In summary, we have obtained the regularization of gradient catastrophe in the KP, upright Y-soliton modulation system by the introduction of another leg and the creation of a locally resonant inverted Y-soliton structure.  This is a combined modulation system for the inverted and upright Y-solitons, which maps to a modulated 3421 soliton solution according to Kodama's classification scheme \cite{kodama_kp_2010}.  See the Appendix \ref{sec:3421_soliton}. 
    This implies that the dimensionality of the system has changed from a 3 parameter family modulated Y-soliton to a 4 parameter family modulated 3421 soliton, with $(K_1, K_2,K_2',K_3)$ characterizing the solution in the vicinity of the two vertices.

This construction is equivalent to the formulation described in \cite{ryskamp_modulation_2022,han_regularizations_2026}, where the Y-soliton can also be interpreted as a shock solution of the modulation equations satisfying Rankine Hugoniot conditions for which certain Riemann invariants exhibit a jump discontinuity at the vertex.  For uniqueness, one must appeal to the admissibility condition:  a modulation shock is admissible if it corresponds to the zero dispersion limit of a traveling wave solution of the KP equation.  This is the analog of the vanishing viscosity limit in conservation law theory for the admissibility of weak solutions to hyperbolic conservation laws \cite{Dafermos2016}. The regularization carried out above is interpreted in this way where the Rankine Hugoniot conditions and admissibility result in the vertex (shock) velocity \eqref{eq:inverted_Y_vertex_speed} and the first continuity condition in  \eqref{eq:1prime_k2_IC}.

    \subsubsection{Asymptotically singular case}
    We next briefly consider when $K_2' =K_3$. That is, consider System \eqref{eq:mod_param_k3}--\eqref{eq:mod_param_k1} subject to initial conditions \eqref{eq:fastRWj3IC} that satisfy 
    \begin{equation}\label{eq:fastRWj3OrderingSingular}
        K_1<K_2<K_3\leq K_2'.
    \end{equation}
     Much of the analysis is analogous to that of Section \ref{sec:leg2AsymptoticSingular}. On leg $j=3$, the parameters $k_1^{(3)} = K_1 < K_2 \le k_2^{(3)} \le K_2'$ are well-ordered.

    For $K_2' \ge K_3$, the first interaction time \(t = t_1\) of the rarefaction with the vertex $y = p(t)$ remains finite and positive.  By comparing the vertex speed $\dot{p}= \frac{2}{3}(K_1 + k_2^{(3)}+K_3)$ with the characteristic velocity $\lambda_+^{(3)}(K_1,k_2^{(3)})$, we observe that $\dot{p} < \lambda_+^{(3)}(K_1,K_2')$ so that the rarefaction wave originating on leg \(j=3\) reaches the vertex, but never fully passes through it. The solution of the modulation equations is the same as that presented in the previous subsection.  The only difference is that $t_2, t_3 \to \infty$ as $K_2' \to K_3$ (see eqs.~\eqref{eq:t1t2New}, \eqref{eq:t3}), which implies that the originating RW never fully passes through $y = p(t)$. Moreover, one can analogously show that the value transmitted to the vertex satisfies
    \[
    k^{(3)}_2(p(t),t)\le K_3,
    \]
    so that the potentially order-violating portion of the rarefaction never reaches leg \(j=1\).

    We briefly discuss the qualitative behavior of the ensuing dynamics, with figure \ref{fig:shockCreationDiagram} providing helpful visual intuition for this discussion.
    A key difference from Section \ref{sec:leg2AsymptoticSingular}, however, is that a new leg is formed in this case. Thus, when \(k_2^{(3)} \to K_2'=K_3\), the long-time asymptotic state is not just the single line soliton composed of legs $j=2$ and $j=3$, each with parameters $(K_1,K_3)$, but also includes a partial line soliton composed of legs $j = 1$ and $j = I$, with parameters $(K_2, K_3)$.  Leg $j=1'$ vanishes.  If $K_2' > K_3$, then this asymptotic separation occurs but now, legs $j=2$ and $j=3$ correspond to a modulated line soliton with the vertex velocity $\dot{p}$ degenerating to the characteristic velocity $\lambda_+^{(3)}(K_1,K_2')$.  See Figure \ref{fig:newLegSpecial} for a numerical simulation of this case.
    \begin{figure}
        \includegraphics[width=0.98\linewidth]{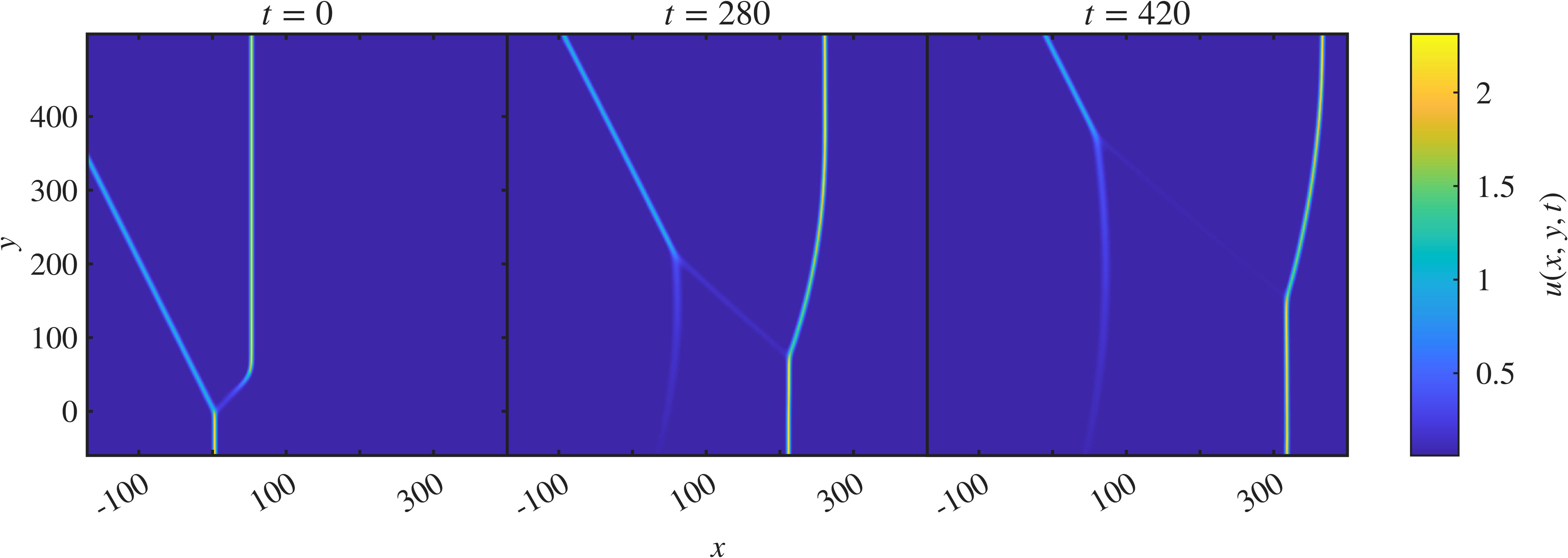}
        \caption{This figure displays the time evolution from a numerical simulation of a fast rarefaction starting from leg $j=3$ and transmitting to leg $j=1$. The relevant parameters are $K_1=-1.5,K_2=-0.5, K_2'=1.5,K_3=1.5$, and $y_0=50$. The amplitude of of leg $j=1'$ becomes comparable to background noise in the simulation.}
        \label{fig:newLegSpecial}
    \end{figure}

    \subsection{Summary of Y-soliton rarefaction modulation solutions}

    Table \ref{tab:transmit_receive} summarizes the six distinct rarefaction solutions, only three of which pass through the Y-soliton vertex $y = p(t)$.  The characteristic family (fast or slow), the transmitting/receiving legs, and whether a new leg is formed are indicated.
    \begin{table}[h!]
        \centering
        \begin{tabular}{|c|c|c|c|}
        \hline
         & \multicolumn{3}{c|}{\textbf{Receiving Leg}} \\ 
        \hline
        \textbf{Transmitting Leg} & 1 & 2 & 3 \\ 
        \hline
        1 & slow \& fast & N/A & N/A \\ 
        \hline
        2 & fast & slow & N/A \\ 
        \hline
        3 & fast with new leg & slow & N/A \\ 
        \hline
        \end{tabular}
        \caption{Summary of rarefaction wave solutions for a modulated Y-soliton. For a given transmitting leg, the table identifies the characteristic family (slow or fast) and receiving leg of the rarefaction solution.}
        \label{tab:transmit_receive}
    \end{table}

\section{Modulations of X-soliton traveling wave solutions}
\label{sec:xSoliModulations}
In this section, we will classify the rarefaction wave solutions for the O-type and P-type X-solitons (cf.~Section \ref{sec:x_soliton}). 

We can formulate the modulation system in a similar manner to how we formulated the Y-soliton system in Sec.~\ref{sec:modul-y-solit} for which each leg is modulated separately, with the vertex serving as a boundary between them.  However, the O-type and P-type X-solitons conserve both phase parameters across the vertex, as opposed to conserving just one, rendering the boundary cosmetic.  In short, the modulations simply pass through the vertex and continue on the leg with the shared phase parameters.  That is, the modulation system for X-solitons can be treated as two independent modulation systems for line solitons, with a trivial vertex essentially serving as a label for the intersection of the line solitons.

Consequently, the modulation equations themselves do not introduce a new coupled characteristic structure at the vertex. For this reason, the linearized well-posedness of this system is trivial, as there will always be four characteristics flowing into the vertex and four characteristics flowing out. The essential qualitative behavior can instead be understood by comparing the characteristic speeds of the composing line solitons with the velocity of their intersection. In a moving frame with the vertex, these relative velocities determine the directions of information propagation along the different legs.

While the mathematical solutions can be described by eqs.~\eqref{eq:1-RW} and \eqref{eq:2-RW}, there are some interesting qualitative distinctions, which we will describe below. 

For strictly ordered phase parameters $(k_1<k_2<k_3<k_4)$, the O-type vertex velocity always lies between the two characteristic velocities associated with the $[k_1,k_2]$ and $[k_3,k_4]$ line solitons. Thus, the characteristic ordering relative to the vertex remains fixed. Furthermore, note the vertex velocity given in \eqref{eq:O_type_yvel} is always finite, which prevents the geometric degeneracy associated with parallel line solitons. See Figure \ref{fig:OTypeMod} for a direct numerical simulation of an O-type modulation passing through the vertex.

In contrast, the P-type vertex velocity does not possess a fixed ordering relative to the characteristic velocity of its composing line soliton. Varying the phase parameters can produce different characteristic configurations relative to the vertex. Furthermore, the vertex velocity defined in \eqref{eq:P_type_yvel} becomes singular when $k_1+k_4=k_2+k_3$, corresponding geometrically to the composing line solitons becoming parallel. In this limit, the intersection of the two line solitons is lost, corresponding to the two line solitons separating from one another. See Figure \ref{fig:PTypeParallel} for a direct numerical simulation of this case. Note that this singularity reflects a degeneracy of the geometric vertex description rather than a singularity of the underlying KPII solution, and should be interpreted as the vertex disappearing, rather than a singularity that needs to be regularized.

Lastly, another notable qualitative behavior is the deletion of a given leg, corresponding to the phase parameters coalescing with one another, causing the amplitude to vanish.  For the O-type soliton, either the $[k_1,k_2]$ or $[k_3,k_4]$ soliton can be deleted by taking $(k_1\to k_2)$ or $(k_3\to k_4)$, respectively. The corresponding deletion boundaries propagate at the limiting characteristic velocity $(s_{\mathrm{del}}=2k_*)$, which is always less than the vertex velocity for the $[k_1,k_2]$ deletion and greater for the $[k_3,k_4]$ deletion See Figure \ref{fig:OTypeDel} for a direct numerical simulation of a deletion propagating through the vertex. For the P-type soliton, only the inner $[k_2,k_3]$ soliton can be independently deleted while respecting the ordering of the phase parameters, since $(k_1\to k_4)$ would require all four parameters to coalesce. Unlike the O-type case, the relative propagation direction of the P-type deletion boundary depends on the phase parameters: it travels slower than the vertex when $k_2-k_1>k_4-k_3$, and faster when the inequality is reversed. The transition between these two cases occurs at $k_1+k_4=k_2+k_3$, where the P-type vertex velocity becomes singular as the constituent solitons become parallel. Thus, deletion boundaries exhibit fixed relative propagation directions for O-type solitons, but parameter-dependent directions for P-type solitons. 

This formulation of X-solitons as independent line soliton modulation system is supported by direct numerical simulations of X-soliton modulations in KP II.  See Figs.~\ref{fig:OTypeMod}, \ref{fig:PTypeParallel}, and \ref{fig:OTypeDel},
which exhibit good agreement with the predicted geometry.

\begin{figure}
    \centering
    \includegraphics[width=0.98\linewidth]{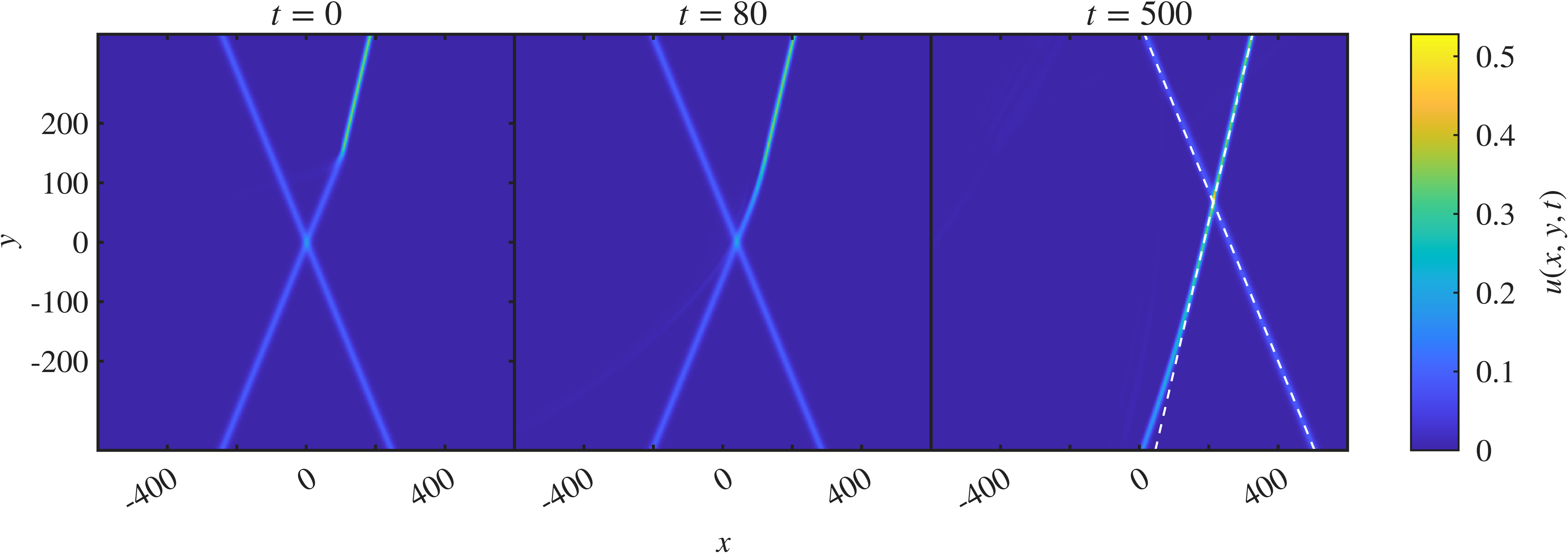}
    \caption{Numerical simulation of an O-type soliton subject to a fast rarefaction wave modulation. The dashed lines show the expected limiting cases of two line soliton modulations with the same parameters.}
    \label{fig:OTypeMod}
\end{figure}

\begin{figure}
    \centering
    \includegraphics[width=0.98\linewidth]{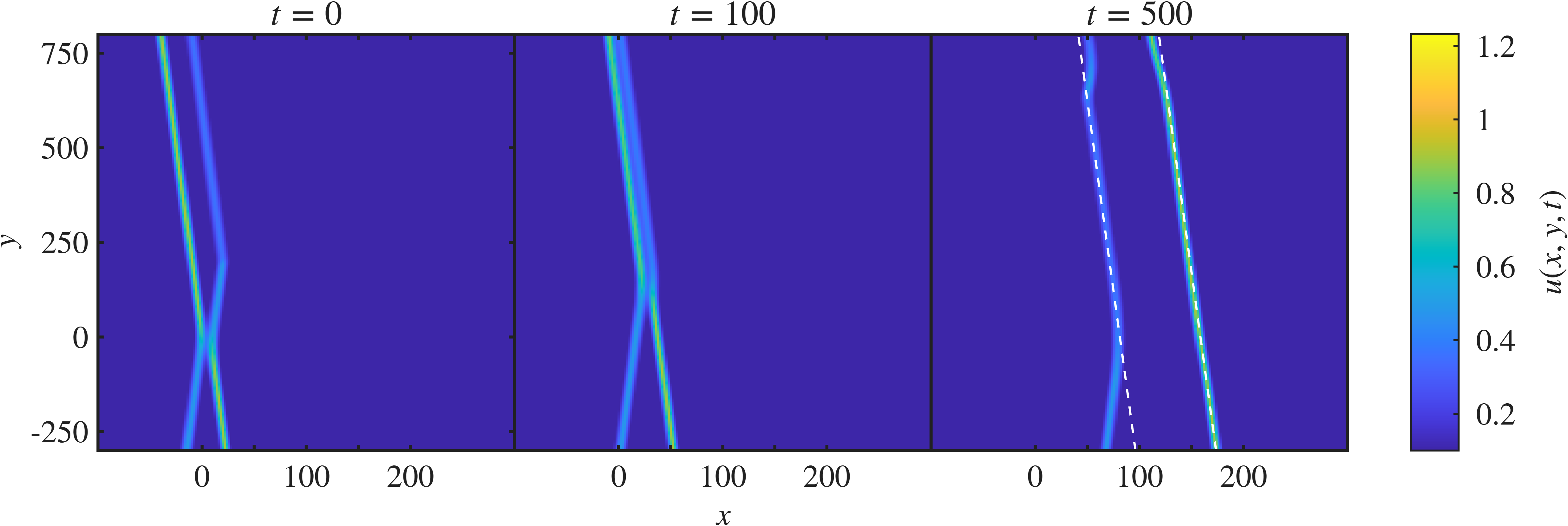}
    \caption{Numerical simulation of a P-type soliton subject to a slow rarefaction modulation that results in the separation of its constituent legs. The dashed line shows the expected limiting cases the resulting separated line soliton.}
    \label{fig:PTypeParallel}
\end{figure}

\begin{figure}
    \centering
    \includegraphics[width=0.98\linewidth]{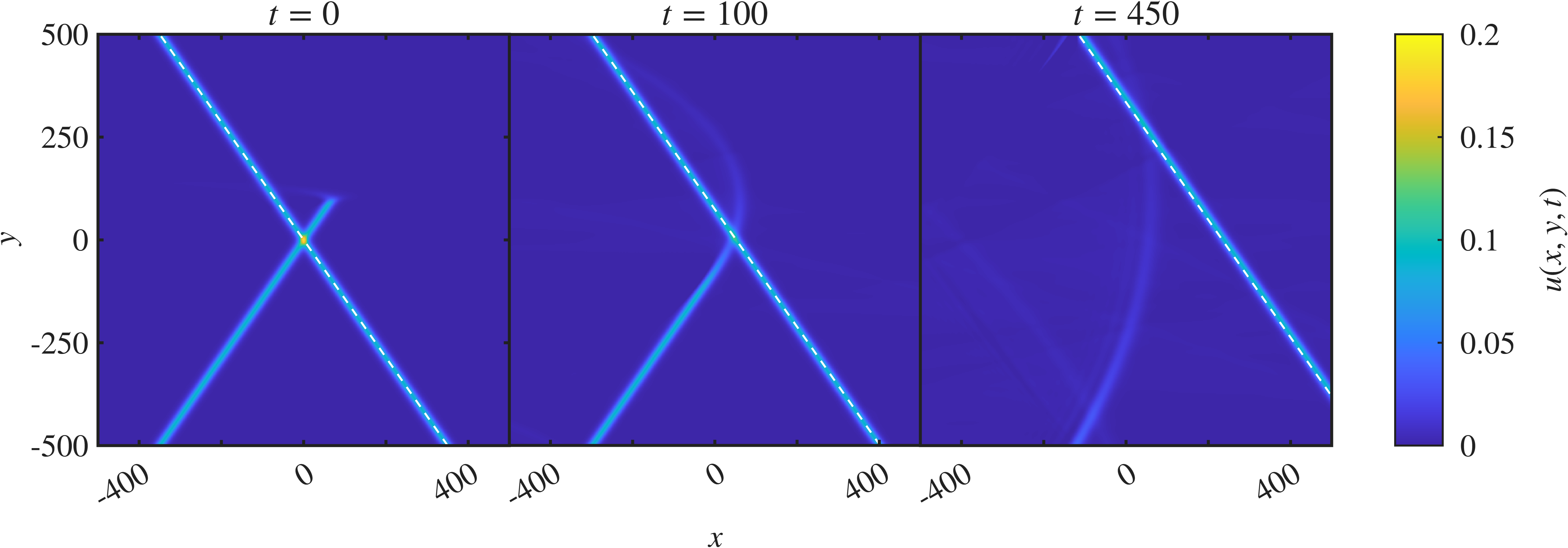}
    \caption{Numerical simulation of an O-type soliton subject to a fast rarefaction wave that leads to leg deletion. The dashed line shows the expected limiting cases the resulting single line soliton.}
    \label{fig:OTypeDel}
\end{figure}

\section{Numerical Simulations \& Comparisons}
\label{sec:direct-numer-simul}

To verify our analytical results, the KP II equation is numerically solved with a pseudospectral method described in \cite{kao_numerical_2012}. An integrating factor,  fourth-order Runge Kutta time stepper is used. Specifically, a Gaussian windowing function is used in order to solve for non-periodic data in $y$, which requires estimates of the outgoing line soliton solutions at the boundary.
The time evolution plots seen throughout Section \ref{sec:ySoliModulations} and \ref{sec:xSoliModulations} were integrated in accordance to this scheme.

In this section, we compare the theoretical predictions with direct numerical simulations of KP II. We focus on the fast rarefaction-wave case from leg $j=3$ leading to new leg formation, since it is the most delicate regime numerically. The remaining cases are more regular and show similarly good or better agreement with the theory.  The modulation solutions are evaluated explicitly when available and otherwise by characteristic tracing. From the direct numerical simulation of KP II, the local amplitude and slope information from each of the legs were extracted to compare with the theoretical modulation solutions.

\subsection{Numerical comparison for fast-RW from leg $j=3$}
\begin{figure}
    \centering
    \includegraphics{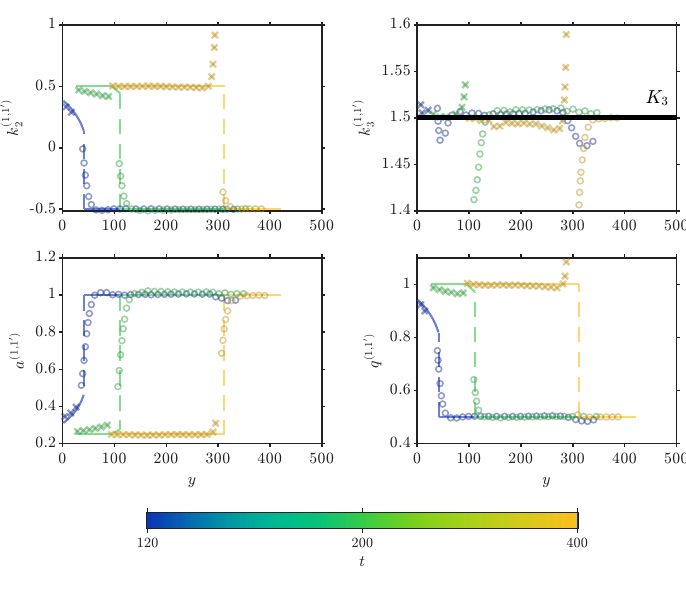}
    \caption{Analytical (solid curves) and extracted from KP II simulation ($\times$ and $\circ$) representations of modulation parameters on legs $j=1$ ($y > s(t)$) and $j=1'$ ($y < s(t)$). The dashed vertical lines represent the newly formed Y-soliton vertex $y = s(t)$, corresponding to a jump discontinuity in the parameters from $k^{(1)}_2$ to $k^{(1')}_2$ while the other phase parameter $k_3^{(1,1')} = K_3$ is conserved. The numerically extracted parameters are denoted by $\times$ and $\circ$ on legs $j =1$ and $j= 1'$, respectively. Color corresponds to time $t$.  The analytical modulation solution has been shifted in $y$ by $10$ to account for short-time, transient dynamics not captured by leading order Whitham theory.}
    \label{fig:sharedLegComparison}
\end{figure}

\begin{figure}
    \centering
    \includegraphics{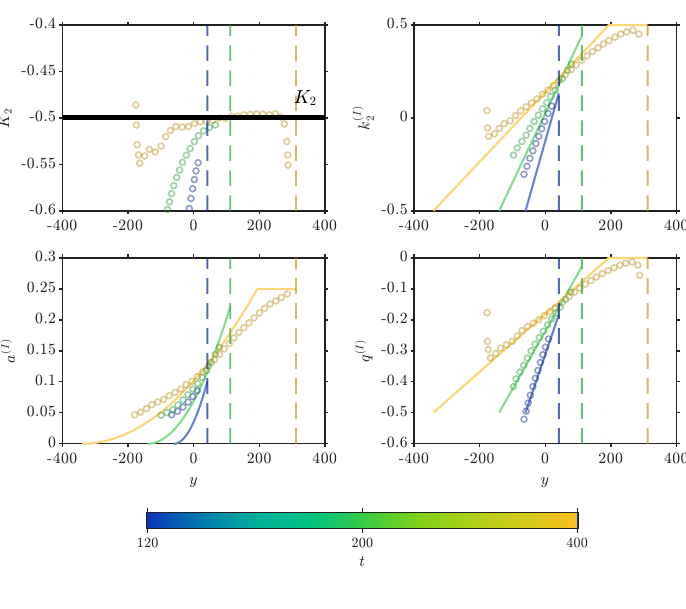}
    \caption{Analytical (solid curves) and extracted from KP II simulation ($\circ$) representations of modulation parameters on leg $j=I$ ($y < s(t)$). The dashed vertical lines represent the newly formed Y-soliton vertex $y = s(t)$. Color corresponds to time $t$.  The analytical modulation solution has been shifted in $y$ by $10$ to account for short-time, transient dynamics not captured by leading order Whitham theory.}
    \label{fig:newLegComparison}
\end{figure}

Figures \ref{fig:sharedLegComparison}--\ref{fig:newLegComparison} show the theoretically predicted and numerically extracted parameters on legs $j=1, 1', I$ during the evolution resulting from a fast-RW initiated on leg $j=3$ with initial modulation parameters $K_1=-1.5,K_2=-0.5, K_2'=0.5,K_3=1.5$, and $y_0=50$.  See section \ref{sec:j3RW} for the corresponding analysis.

On legs $j=1,1'$, the analytical prediction and numerical simulation are in good agreement away from the vertex $s(t)$, with $\dot{s}(t)$ accurately predicted.  The numerical results exhibit oscillations near $y = s(t)$, which is expected in this dispersively regularized asymptotic theory. The conserved phase parameter $k_3$ is difficult to precisely extract from numerical simulation in the vicinity of the vertex $y = s(t)$.

On the new leg $j=I$, the predicted amplitude consistently underestimates the numerically extracted amplitude, especially at early times, causing the reconstructed phase parameters to vary more. This may be due to the impact of dispersive waves generated during gradient catastrophe and the small amplitude of the generated leg. This is in contrast to the accurate slope predictions for all times, which are more robust to higher order amplitude perturbations.

Overall, the numerical simulations show good agreement with the analytical predictions, especially for long time.  Similar agreement is obtained for the other rarefaction wave modulation configurations discussed in this paper.  For brevity, we have focused on the most complex case involving gradient catastrophe for quantitative comparison to demonstrate the effectiveness of the theory.  The reported quantitative agreement supports the assumed structure of the modulated Y-soliton as three distinct, conjoined semi-infinite line solitons with a shared boundary given by the Y-soliton vertex and suitable continuity conditions.

\section{Discussion \& Conclusions}
\label{sec:conclusions}
This paper has examined the dynamical stability and nonlinear modulation dynamics of Y- and X-wave KP II solitons under long-wavelength perturbations. At the linearized level, the analysis indicates dynamical stability of the soliton structures, as reflected by the exclusively real characteristic velocities of the modulation system. Beyond the linear analysis, special rarefaction-wave solutions to the full modulation equations were analyzed. 

In the case of the Y-soliton, interesting dimensionality and structure altering behaviors were observed, namely the deletion of a soliton leg and the creation of a new leg through gradient catastrophe. Certain choices of modulation parameters lead to the asymptotic reduction of the Y-soliton into a modulated line soliton, corresponding to a transition from a three-parameter family to a two-parameter family of waves. Another family of modulation data lead to wave breaking in the Y-soliton modulation system. Following wave breaking, an augmented modulation system was introduced to regularize gradient catastrophe, providing an explicit analytical description of new leg formation.  This regularization results in another local Y-soliton vertex structure. In particular, for the fast rarefaction wave emanating from leg $j=3$, the dynamics effectively enlarge the original three-parameter family of Y-soliton solutions to a four parameter family, with $(K_1,K_2,K_2',K_3)$ all contributing to the long-time state. The asymptotic configuration in this case approaches a $(3421)$-type soliton in the classification of \cite{kodama_kp_2010}. These structural changes are mirrored by changes in the soliton graph and the regions of influence of the phase parameters. 

On the other hand, it was found that the vertices in X-solitons are trivial. That is, since both phase parameters are conserved across X-soliton vertices, the X-soliton modulation systems can be decoupled into two independent line soliton modulation systems. In this case, the vertex serves as a tracker for interesting qualitative behavior of these modulations, namely the possibility of separation of the two line solitons for the P-type soliton.

In both modulated Y-soliton and modulated X-soliton cases, the theoretical predictions are supported by numerical simulations, which exhibit good quantitative agreement in regimes of sufficiently large amplitude.

Recent results in \cite{Mizumachi2026} on spectral stability of the Y-soliton solution of KP II provide a rigorous foundation for our postulated description of a perturbed Y-soliton.  Namely, the tri-legged, independently modulated structure with a shared boundary and suitable continuity conditions.  The modulation characteristic velocities presented here are precisely the long wavelength limits of the group velocities developed in \cite{Mizumachi2026}.  Our approach, however, goes well beyond the linearized perturbation level and allows for the propagation of nonlinear modulations and their geometric implications.

Taken together, these findings may be used to describe more complex KP II soliton structures such as web-type solutions \cite{kodama_kp_2010}. Local reductions of the soliton graph to either resonant Y-soliton or trivial X-soliton vertices can be made in order to predict how a modulation evolves along the soliton graph.  The 3421-soliton that emerges during the course of a particular Y-soliton modulation serves as a representative example of this strategy.  Furthermore, these findings reinforce the importance and generality of the resonant Y-soliton vertex in the KP II system from a modulational perspective and, more broadly, in applications such as shallow water waves \cite{ablowitz_nonlinear_2012}. 
Not only can a shock in the line soliton modulation system produce a Y-resonant structure \cite{ryskamp_modulation_2022, han_regularizations_2026}, but as found in this paper, a family of initially expansive wave initial data to the Y-soliton modulation can lead to the creation of a new, locally Y-resonant structure.

\backmatter

\bmhead{Acknowledgements}

The authors are grateful to Sarbarish Chakravarty and Gino Biondini for helpful discussions and intuition on KP II multi-soliton solutions.  The authors acknowledge support from NSF-DMS 2306319.

\begin{appendices}

\section{Gradient catastrophe in Y-soliton modulations}
\label{sec:appendix_grad_catastrophe}

The following theorem and corollary are described and utilized in Section \ref{sec:modul-y-solit}.

    \begin{theorem} \label{th:shockCondition}
        Let $k=k(y,t)$ satisfy
        \begin{equation} \label{eq:generalShockEquationAnalysis}
            \frac{\partial k}{\partial t} + c(k)\frac{\partial k}{\partial y} = 0 \quad y>p(t) \text{ or } y<p(t),
        \end{equation}
        with moving boundary \(y=p(t)\) and boundary condition
        $k(p(t), t) = f(t).$  Furthermore, assume $c \in C^1(\mathbb{R})$ with $c'(k)>0$, 
        $p,f \in C^1(I)$ on an open interval $I\subset\mathbb{R}$ with bounded  velocity $\dot{p}$
        \[
        V_- \le \dot{p}(t) \le V_+ \quad  t\in I,
        \]
        for some $V_\pm \in \mathbb{R}$.
         If the domain of equation \eqref{eq:generalShockEquationAnalysis} is $y>p(t)$. Gradient catastrophe occurs if there exists a $t^* \in I$ such that
        \begin{equation}\label{eq:aboveVertexShockCondition}
         f'(t^*) > 0\;\;\text{and}\;\;c(f(t^*)) > V_+,
        \end{equation}
        or, if the domain of equation \eqref{eq:generalShockEquationAnalysis} is $y<p(t)$
        \begin{equation}\label{eq:belowVertexShockCondition}
         f'(t^*) < 0\;\;\text{and}\;\;c(f(t^*)) <V_-.
        \end{equation}
        In both cases, the breaking time $t_{\text{break}}$ is
        \begin{subequations}\label{eq:shockTime}
            \begin{align}
            &T(t^*) = t^* + \frac{c(f(t^*))-\dot{p}(t^*)}{c'(f(t^*))\, f'(t^*)}, \\
            &t_{\text{break}} = \inf_{t^*\in A} T(t^*).
            \end{align}
        \end{subequations}
        where $A$ is the set of $t^*$ that satisfies conditions \eqref{eq:aboveVertexShockCondition}, \eqref{eq:belowVertexShockCondition} respectively.
    \end{theorem}
    \begin{proof} 
        We can prove this by simply calculating the gradient catastrophe time if it exists. 
        Using the method of characteristics, the solution of eq.~\eqref{eq:generalShockEquationAnalysis} is
        \[
        k(y,t) = f\left(t^*(y,t)\right).
        \]
        So
        \begin{equation}\label{eq:dkdyDefinition}
            \frac{\partial k}{\partial y} = \frac{df}{dt^*}\frac{\partial t^*}{\partial y},
        \end{equation}
        where the implicit definition of $t^*$ is 
        \begin{equation}\label{eq:ytstarCharDefinition}
            y=c(f(t^*))\cdot(t-t^*) + p(t^*)
        \end{equation}
        Implicitly differentiating with respect to $y$ and solving for $\frac{\partial t^*}{\partial y}$ yields
        \[
         \frac{\partial t^*}{\partial y}=\frac{1}{c'(f(t^*))\cdot f'(t^*)\cdot(t-t^*)-c(f(t^*))+\dot{p}(t^*)}
        \]
        Inserting this into relation \eqref{eq:dkdyDefinition} yields
        \[
        \frac{\partial k}{\partial y} = \frac{f'(t^*)}{c'(f(t^*))\cdot f'(t^*)\cdot(t-t^*)-c(f(t^*))+\dot{p}(t^*)}
        \]
        Now, provided that $f'(t^*) \neq 0$, a singularity appears when 
        \[
        c'(f(t^*))\cdot f'(t^*)\cdot(t_\text{break}-t^*)-c(f(t^*))+\dot{p}(t^*) = 0.
        \]
        Note that $t^*$ may become multivalued after gradient catastrophe. For the purposes of this calculation, we consider only characteristics originating from the moving boundary $y=p(t)$, rather than from the initial line $t=0$. Specifically, $t^*$ denotes the time at which a characteristic originates from $(p(t^*),t^*)$, and we require that the characteristic remains in the prescribed domain for $t^*<t\leq T(t^*)$. That is,
        \[
        c(f(t^*))(t-t^*)+p(t^*)>p(t)
        \]
        for $y>p(t)$, with the inequality reversed for $y<p(t)$. Note that the characteristics themselves remain parametrically defined by their originating times $t^*$, allowing us to determine the candidate breaking time without requiring $t^*(y,t)$ to remain single-valued beyond breaking.

        Solving for our set of proposed breaking times $T$ and our breaking time $t_\text{break}$ yields
        \begin{equation}
           \begin{split}
            &T(t^*) =  t^*+\frac{c(f(t^*))-\dot{p}(t^*)}{c'(f(t^*))\cdot f'(t^*)}\\
            &t_{\text{break}} = \inf_{t^*\in A} T(t^*).
            \end{split}
            \label{eq:breaking_time}
        \end{equation}
        Let's denote $t^*_b \in A$ as the corresponding value of $t^*$ that achieves the breaking time if it exists. 
        
        Now we have to enforce that the breaking time is physical, that is we require $t_\text{break} - t^*_b >0$ and that the break position $y_\text{break}$ lies within the domain $y_\text{break} - p(t_\text{break}) >0$ for the case of \eqref{eq:aboveVertexShockCondition} and $y_\text{break} - p(t_\text{break}) < 0$ for the case of \eqref{eq:belowVertexShockCondition}. Enforcing the positivity of $t_\text{break} - t^*_b$ yields the condition
        \[
            \frac{c(f(t^*_b))-\dot{p}(t^*_b)}{c'(f(t^*_b))\cdot f'(t^*_b)}>0,
        \]
        and along with the assumption that $c'(k) >0$ this yields
        \begin{equation}\label{eq:positiveBreakingCondition}
            \frac{c(f(t^*_b))-\dot{p}(t^*_b)}{f'(t^*_b)}>0.
        \end{equation}
        Now we enforce the constraint that $y_\text{break}$ must lie within the domain.
        To determine whether the breaking point lies inside the domain, note that the breaking location is, by the method of characteristics,
        \[
        y_\text{break}
        =
        p(t^*_b) + c(f(t^*_b))\cdot(t_{\mathrm{break}}-t^*_b),
        \]
        while we also note that
        \[
        p(t_{\mathrm{break}})  =p(t^*_b)+ \int_{t^*_b}^{t_{\mathrm{break}}}\dot{p}(s) ds.
        \]
        Thus, we can write the difference as
        \[
        y_{\mathrm{break}}-p(t_{\mathrm{break}})
        =
        \int_{t^*_b}^{t_{\mathrm{break}}}\bigl(c(f(t^*_b))-\dot{p}(s)\bigr)\,ds.
        \]
        Since we required $t_{\mathrm{break}}-t^*_b>0$ to determine the sign of $y_{\mathrm{break}}-p(t_{\mathrm{break}})$, it is sufficient to show that $c(f(t^*_b))-\dot{p}(s)$ has the correct sign on the interval $[t^*_b, t_\text{break}]$.

        For the case of \eqref{eq:aboveVertexShockCondition}, we require $y_{\mathrm{break}}-p(t_{\mathrm{break}})>0$, which is satisfied when 
        \[
         c(f(t^*_b)) > V_+>\dot{p}(t^*_b)\implies c(f(t^*_b))-\dot{p}(s)>\left(c(f(t^*_b))-V_+\right)>0 .
        \]
        This constraint is sufficient to prove $y_{\mathrm{break}}-p(t_{\mathrm{break}})>0$, i.e., that $y_{\rm break}$ is in the domain.  Combining this with eq.~\eqref{eq:positiveBreakingCondition}, we have $f'(t^*)>0$ and the theorem is proved for this case.
        
        Similar for the case of \eqref{eq:belowVertexShockCondition}, we require $y_{\mathrm{break}}-p(t_{\mathrm{break}})<0$, which is satisfied when
        \[
        c(f(t^*_b)) <V_-< \dot{p}(t^*_b) \implies c(f(t^*_b))-\dot{p}(s)<c(f(t^*_b))-V_-<0.
        \]
        This combined by \eqref{eq:positiveBreakingCondition} implies $f'(t^*_b)<0$, which yields the stated conditions for \eqref{eq:belowVertexShockCondition}. \qed

        \begin{corollary}
        \label{cor:noShockCriterion}
        The condition that either $f'(t) > 0$, $t \in I$ for the domain $y > p(t)$ or $f'(t) < 0$, $t \in I$ for the domain $y < p(t)$ is necessary for gradient catastrophe to occur. 
        
        \end{corollary}

        \begin{proof}
        We prove the case for \(y>p(t)\) in detail, with the case \(y<p(t)\) being analogous.
        
        We fix the boundary emission time \(t^*\in I\). The characteristic emanating from the boundary at $t = t^*$ is
        \[
        y=c(f(t^*))(t-t^*)+p(t^*),
        \]
        with $y$ derivative
        \[
        \frac{\partial k}{\partial y}
        =
        \frac{f'(t^*)}
        {D(t)}, \quad D(t) = c'(f(t^*))\,f'(t^*)\,(t-t^*)-c(f(t^*))+\dot p(t^*).
        \]
        Immediately after emission, we have
        \[
        c(f(t^*))-\dot p(t^*)>0,
        \]
        so that
        \[
        D(t^*)=-c(f(t^*))+\dot p(t^*)<0.
        \]
        Anticipating a contradiction, we assume \(f'(t^*)\le 0\). Since \(c'(f(t^*))>0\),
        \[
        D'(t)=c'(f(t^*))\,f'(t^*)\le 0.
        \]
        From this, we see that \(D(t)\) is nonincreasing and starts negative at \(t=t^*\). Thus \(D(t)\) cannot vanish for any \(t>t^*\), showing that \(\partial k/\partial y\) cannot blow up in finite forward time.
        
        Since this analysis applies to every boundary emitted characteristic, no gradient catastrophe can occur in the region \(y>p(t)\).
        
        The proof for the region \(y<p(t)\) is the same, except that
        \[
        c(f(t^*))-\dot p(t^*)<0,
        \]
        so \(D(t^*)>0\), and if \(f'(t^*)\ge 0\), then \(D'(t)\ge 0\), so again \(D(t)\) cannot cross zero for \(t>t^*\).
        \end{proof}
        Theorem \ref{th:shockCondition} gives sufficient conditions for gradient catastrophe to occur, while Corollary \eqref{cor:noShockCriterion} shows the necessity of the condition of $f'(t^*) > 0$ and $f'(t^*) < 0$ for gradient catastrophe to occur when equation \eqref{eq:generalShockEquationAnalysis} is defined for $y>p(t)$ and $y<p(t)$ respectively. 
    \end{proof}

When we apply Corollary \eqref{cor:noShockCriterion}, we see that the fast-RW from leg $j=2$ to leg $j=1$ and the slow-RW from leg $j=3$ to leg $j=1$ will fail to exhibit gradient catastrophe.  However, when solving for the fast-RW from leg $j=3$ to leg $j=1$, sufficiency for gradient catastrophe is met, as specified in Theorem \ref{th:shockCondition}.

\section{3421 soliton}\label{sec:3421_soliton}

The 3421 soliton under Kodama's classification in \cite{kodama_kp_2010} corresponds to another $(M,N)=(4,2)$ parameter solution to the multi-soliton in eq.~\eqref{eq:NSolitonSolution} parameterized by $k_1 < k_2 < k_3 < k_4$ and coefficient matrix
\begin{equation}\label{eq:3421TypeMatrix}
    A = \begin{pmatrix}
        1 & 0 & 0 &-\gamma\\
        0 & 1 & \alpha & \beta
    \end{pmatrix} \quad \alpha,\beta,\gamma >0.
\end{equation}
Unlike the previously discussed examples, the 3421 soliton is not a traveling wave, with its structure varying over time. See Figure \ref{fig:3421} for a 3421 soliton evaluated at two separate times.
\begin{figure}
    \centering
    \includegraphics[clip=true,trim=0 25 0 25]{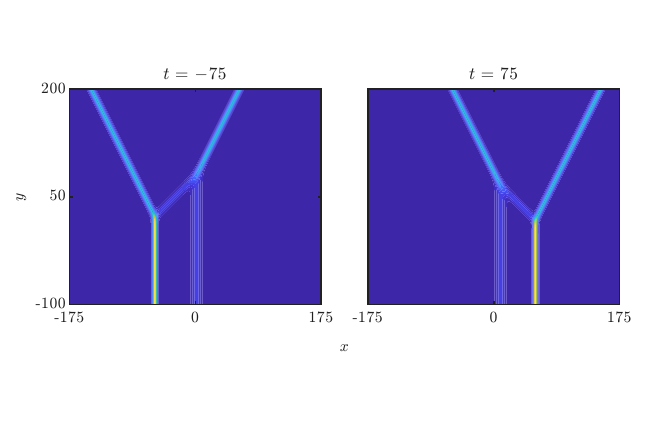}
    \caption{A 3421 soliton with parameters $(k_1,k_2,k_3,k_4) = (-1.5, -0.5, 0.5, 1.5)$ evaluated at $t=-75$ and $t=75$.}
    \label{fig:3421}
\end{figure}
Note that asymptotically, the 3421 soliton is composed of two-resonant Y-soliton vertices.

\end{appendices}

\bibliographystyle{sn-mathphys-num}
\bibliography{KP_refs}%

\end{document}